\documentclass[journal,10pt]{IEEEtran}
\usepackage{amsmath}
\usepackage{amssymb}
\usepackage{amsfonts}
\usepackage{graphicx}
\usepackage{epsfig}
\usepackage{subfigure}
\usepackage{psfrag}
\usepackage{color}
\usepackage{url}
\usepackage{cite}
\begin{document}
\title{Multiuser MISO Beamforming for Simultaneous Wireless Information and Power Transfer}

\author{Jie Xu, Liang Liu, and Rui Zhang
\thanks{This paper was presented
in part at the IEEE International Conference on Acoustics, Speech, and Signal
Processing (ICASSP), Vancouver, Canada, May 26-31, 2013 \cite{XuLiuZhangConf}.}
\thanks{J. Xu and L. Liu are with the Department of
Electrical and Computer Engineering, National University of
Singapore (email: elexjie@nus.edu.sg, liu\_liang@nus.edu.sg).}
\thanks{R. Zhang is with
the Department of Electrical and Computer Engineering, National
University of Singapore (e-mail: elezhang@nus.edu.sg). He is also
with the Institute for Infocomm Research, A*STAR, Singapore.}}
\maketitle

\begin{abstract}\label{sec:abstract}
Simultaneous wireless information and power transfer (SWIPT) is anticipated to {\color{black}{have abundant applications in future wireless networks}} by providing wireless data and energy access at the same time. In this paper, we study a multiuser multiple-input single-output (MISO) broadcast SWIPT system, where a multi-antenna access point (AP) sends wireless information and energy simultaneously via spatial multiplexing to multiple single-antenna receivers each of which implements information decoding (ID) or energy harvesting (EH). We aim to maximize the weighted sum-power transferred to all EH receivers subject to a given set of minimum signal-to-interference-and-noise ratio (SINR) constraints at different ID receivers. In particular, we consider two types of ID receivers (referred to as Type I and Type II, respectively) without or with the capability of cancelling the interference from ({\emph{a priori}} known) energy signals. For each type of ID receivers, we formulate the joint information and energy transmit beamforming design as a non-convex quadratically constrained quadratic program (QCQP). First, we obtain the globally optimal solutions for our formulated QCQPs by applying an optimization technique so-called semidefinite relaxation (SDR). It is shown via SDR that {\color{black}under the condition of independently distributed user channels,} no dedicated energy beam is used for the case of Type I ID receivers to achieve the optimal solution; while for the case of Type II ID receivers, employing no more than one energy beam is optimal. Next, in order to provide further insight to the optimal design, we establish a new form of the celebrated uplink-downlink duality for our studied downlink beamforming problems, and thereby develop alternative algorithms to obtain the same optimal solutions as by SDR. Finally, numerical results are provided to evaluate the performance of proposed optimal beamforming designs for MISO SWIPT systems, as compared to other heuristically designed schemes.
\end{abstract}


\IEEEpeerreviewmaketitle
\setlength{\baselineskip}{1\baselineskip}
\newtheorem{definition}{\underline{Definition}}[section]
\newtheorem{fact}{Fact}
\newtheorem{assumption}{Assumption}
\newtheorem{theorem}{\underline{Theorem}}[section]
\newtheorem{lemma}{\underline{Lemma}}[section]
\newtheorem{corollary}{\underline{Corollary}}[section]
\newtheorem{proposition}{\underline{Proposition}}[section]
\newtheorem{example}{\underline{Example}}[section]
\newtheorem{remark}{\underline{Remark}}[section]
\newtheorem{algorithm}{\underline{Algorithm}}[section]
\newcommand{\mv}[1]{\mbox{\boldmath{$ #1 $}}}

\section{Introduction}\label{sec:introduction}
Energy harvesting from the environment is a promising solution to provide cost-effective and perpetual power supplies for wireless networks. Besides other well known environmental sources such as wind and solar power, ambient radio signals is a viable new source for energy harvesting. On the other hand, radio signals have been widely used for wireless information transmission. As a result, a unified study on simultaneous wireless information and power transfer (SWIPT) has recently drawn significant attention, which is not only theoretically intricate but also practically valuable for enabling both the wireless data and wireless energy access to mobile terminals at the same time.

There have been a handful of prior studies on SWIPT in the literature (see e.g. \cite{Varshney2008,GroverSahai2010,LiuZhangChua2012,ZhouZhangHo2012,ZhangHo2012,XiangTao2012,ChaliseZhangAmin2012,JuZhang2013}). In \cite{Varshney2008}, SWIPT in a point-to-point single-antenna additive white Gaussian noise (AWGN) channel was first studied from an information-theoretic standpoint. This work was then extended to frequency-selective AWGN channels in \cite{GroverSahai2010}, where a non-trivial tradeoff between information rate and harvested energy was shown by varying power allocation over frequency. The authors in \cite{LiuZhangChua2012} studied SWIPT for fading AWGN channels subject to time-varying co-channel interference, and proposed a new principle termed ``opportunistic energy harvesting'' where the receiver switches between harvesting energy and decoding information based on the wireless channel condition and interference power level. In \cite{ZhouZhangHo2012}, various practical receiver architectures for SWIPT were investigated, where a new integrated information and energy receiver design was proposed. Moreover, motivated by the great success of multi-antenna techniques in wireless communication, SWIPT for multiple-input multiple-output (MIMO)  channels has been investigated in \cite{ZhangHo2012,XiangTao2012,ChaliseZhangAmin2012,JuZhang2013}. In \cite{ZhangHo2012}, Zhang and Ho first investigated SWIPT for the MIMO broadcast channel (BC) with a multi-antenna transmitter sending information and energy simultaneously to one pair of energy receiver and information receiver, each with single or multiple antennas. Under two practical setups where information and energy receivers are either separated or co-located, the optimal precoder designs were developed to achieve various information and energy transmission tradeoffs. The study in \cite{ZhangHo2012} was also extended to the cases with imperfect channel state information (CSI) at the transmitter in \cite{XiangTao2012} and  MIMO relay broadcast channels in \cite{ChaliseZhangAmin2012}. In \cite{JuZhang2013}, a transmitter design based on random beamforming was proposed for a multiple-input single-output (MISO) SWIPT system with artificial channel fading generated for improving the performance of opportunistic information decoding (ID) versus energy harvesting (EH) over quasi-static channels, when the CSI was not available at the transmitter.

\begin{figure}
\centering
 \epsfxsize=1\linewidth
    \includegraphics[width=8cm]{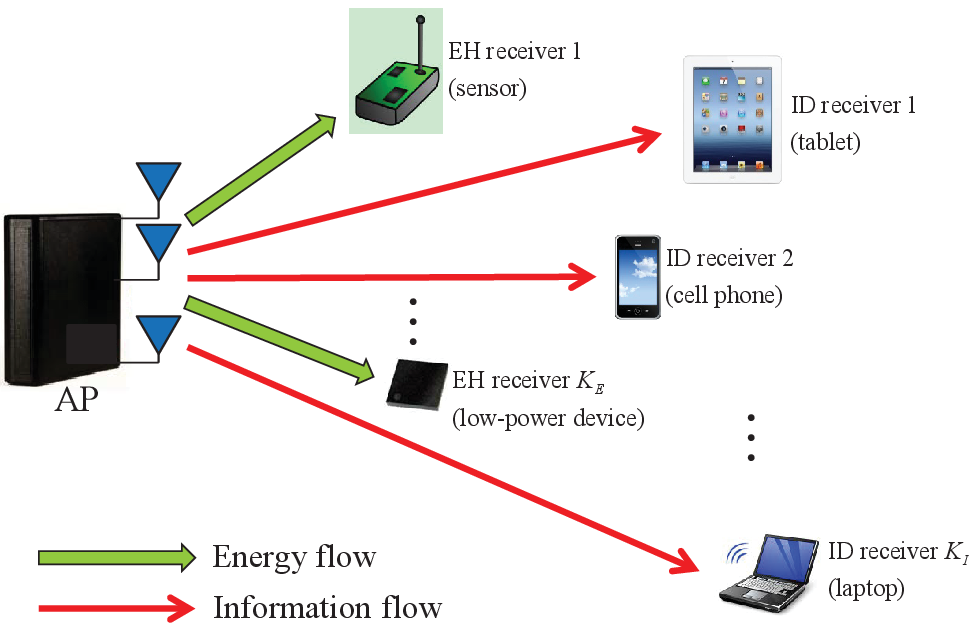}
\caption{A MISO broadcast system for simultaneous wireless information and power transfer (SWIPT), where EH receivers are close to the AP for effective energy reception.} \label{fig:1}
\end{figure}

Despite of the above theoretical advance, in order to implement SWIPT systems in practice, many challenging issues still remain to be addressed. For example, conventional wireless information and energy receivers were separately designed to operate with very different power requirements (e.g., an EH receiver for a low-power sensor requires a received power of $0.1$ mW or $-10$ dBm for real-time operation, while ID receivers such as cellular and Wi-Fi mobile receivers often operate with a received power less than $-50$ dBm \cite{ZhangHo2012}). As a result, existing EH circuits for radio signals are not yet able to be used for ID directly and vice versa. This thus motivates our work in this paper to investigate a practical design of the MISO broadcast system for SWIPT as shown in Fig. \ref{fig:1}, where a multi-antenna access point (AP) transmits simultaneously to multiple single-antenna receivers each of which implements EH or ID, but {\it not both} at the same time. In particular, we consider a receiver location based transmission scheme, where the EH receivers (e.g., sensors and other low-power devices) are deployed sufficiently close to the AP, while the ID receivers (e.g., tablet, cell phone and laptop) can be located more distant from the AP. Notice that the proposed transmission scheme resolves the mismatched power issue for EH and ID receivers as mentioned above, and thus makes the SWIPT system realizable with existing EH and ID receivers. Also note that the location based transmission should be designed in practice by taking into account the potential mobility of receivers to ensure certain fairness in energy and information delivery over time. Under this setup, we aim to jointly design the beamforming weights and power allocation at the transmitter to optimally balance the performance tradeoffs among different information/energy receivers. Specifically, we study the joint information and energy transmit beamforming design to maximize the weighted sum-power transferred to all EH receivers subject to a given set of minimum signal-to-interference-and-noise ratio (SINR) constraints at different ID receivers. In particular, we consider two types of ID receivers, namely {\it Type I} and {\it Type II} receivers, which {\it do not possess} and {\it possess} the capability of cancelling the interference from simultaneously transmitted energy signals ({\color{black}whose waveforms} are assumed to be {\it a priori} known at the transmitter and all Type II ID receivers), respectively. For each type of ID receivers, the design problem is formulated as a quadratically constrained quadratic program (QCQP), which is non-convex and thus difficult to be solved optimally by standard convex optimization techniques \cite{LuoMaSo2010}.

First, we obtain the optimal solutions to the formulated non-convex QCQPs for two types of ID receivers by applying an optimization technique so-called semidefinite relaxation (SDR), {\color{black}{and show that under the condition of independently distributed user channels, the SDRs are tight for the formulated non-convex QCQPs.}} It is revealed that for the case of Type I ID receivers, no dedicated energy beam is used to achieve the optimal solution, while for the case of Type II ID receivers, employing no more than one energy beam is optimal. It is worth noting that SDR has been widely applied in the MISO broadcast channel (see e.g. \cite{LuoMaSo2010,Gershman2010} and the references therein) to obtain efficient (and even optimal under certain conditions) beamforming solutions for various information transmission problems; however, the existing results are not directly applicable to our newly formulated problems with the joint energy and information beamforming optimization.

Next, in order to gain further insight to the optimal joint energy and information beamforming design, we reformulate the QCQP problem for each ID receiver type to an equivalent transmit power minimization problem for the MISO-BC with information transmission only {\color{black}{by leveraging the fact that the SDRs are tight for both QCQPs, based upon which we establish a new form of the celebrated ``uplink-downlink'' duality result}}. By applying the new duality, we develop alternative algorithms based on iterative uplink and downlink transmit optimization to obtain the same optimal downlink beamforming solutions as by SDR. It is worth noting that the uplink-downlink duality has been extensively investigated in the literature to solve non-convex transmit beamforming optimization problems in MISO/MIMO BCs for e.g. SINR balancing in \cite{SchubertBoche2004}, transmit power minimization in \cite{WieselEldar2006,YuLan2007}, and capacity region computation in  \cite{ZhangZhangLiang2012}. Especially, for transmit power minimization in MISO-BC with given SINR constraints for information transfer only, it was shown in \cite{YuLan2007} that the downlink beamforming  problem can be transformed into its dual multiple-access-channel (MAC) problem with an equivalent noise at the receiver characterized by a certain positive semidefinite covariance matrix, which is then solved by applying a fixed-point iteration \cite{Yates1995}. In this paper, this particular type of uplink-downlink duality is extended to the more challenging case when the equivalent noise covariance matrix in the dual MAC is not necessarily positive semidefinite,\footnote{Notice that the covariance matrix of any practical noise cannot be non-positive semidefinite; however, this does not contradict the non-positive semidefinite noise in our case since it is just a mathematical equivalence, and thus needs not be physically realizable.} as a result of the new consideration of joint information and energy transmission, {\color{black}{which renders a non-convex beamforming problem that maximizes a convex (quadratic) objective function. To the authors' best knowledge, the uplink-downlink duality for this new setup has not been studied before.}}

The remainder of this paper is organized as follows. Section \ref{sec:system} introduces the system model and problem formulations. Sections \ref{sec:SDP} and \ref{sec:BC-MAC} present the optimal solutions for the formulated problems based on the approaches of SDR and uplink-downlink duality, respectively. Section \ref{sec:simulation} provides numerical examples to validate our results and compare the performances. Finally, Section \ref{sec:conclusions} concludes the paper.

{\it Notations:} Boldface letters refer to vectors (lower  case) or matrices (upper case). For a square matrix $\mv{S}$, ${\mathtt{tr}}(\mv{S})$ and $\mv{S}^{-1}$ denote  its trace  and inverse, respectively, while $\mv{S}\succeq \mv{0}$, $\mv{S}\preceq \mv{0}$, $\mv{S}\prec \mv{0}$ and $\mv{S}\nsucceq \mv{0}$ mean that $\mv{S}$ is positive semidefinite, negative semidefinite, negative definite and non-positive semidefinite, respectively. For an arbitrary-size matrix $\mv{M}$, ${\mathtt{rank}}(\mv{M})$, $\mv{M}^\dagger$, $\mv{M}^H$, and $\mv{M}^T$ denote the rank, pseudoinverse, conjugate transpose and transpose of $\mv{M}$, respectively, and $\mv{M}_{ik}$ denotes the element in the $i$th row and $k$th column of $\mv{M}$. $\mv{I}$ and $\mv{0}$ denote an identity
matrix and an all-zero matrix, respectively, with appropriate
dimensions. The distribution of a circularly symmetric complex Gaussian (CSCG) random vector with mean vector $\mv{x}$ and covariance matrix $\mv{\Sigma}$ is denoted by $\mathcal{CN}(\mv{x,\Sigma})$; and $\sim$ stands for ``distributed as''. $\mathbb{C}^{x\times y}$ denotes the space of $x\times y$ complex matrices. $\mathbb{R}$ denotes the set of real numbers. ${\mathbb{E}}(\cdot)$ denotes the statistical expectation. $\|\mv{x}\|$ denotes the Euclidean norm of a complex vector $\mv{x}$, and $|z|$ denotes the magnitude of a complex number $z$.  $\rho(\mv{B})$ denotes the spectral radius of a matrix $\mv{B}$, which is defined as the maximum absolute value of the eigenvalues of $\mv{B}$. For two real vectors $\mv{x}$ and
$\mv{y}$, $\mv{x}\ge \mv{y}$ means that $\mv{x}$ is greater than or equal to $\mv{y}$ in a
component-wise manner.

\section{System Model and Problem Formulation}\label{sec:system}

We consider a multiuser MISO downlink system for SWIPT over one single frequency band as shown in Fig. \ref{fig:1}. It is assumed that there are $K_I$ ID receivers and $K_{E}$  EH receivers, denoted by the sets $\mathcal{K_I}=\{1,\ldots,K_I\}$ and $\mathcal{K_E} = \{1,\ldots,K_E\}$, respectively. Also assume that the AP is equipped with $M$ antennas, $M > 1$, and each receiver is equipped with one single antenna. In this paper, we consider linear precoding at the transmitter for SWIPT and each ID/EH receiver is assigned with one dedicated information/energy transmission beam without loss of generality. Hence, the transmitted signal from the AP is given by
\begin{align}
\mv{x} = \sum\limits_{i\in {\mathcal{K_I}}}{\mv w}_i s_i^{\rm{ID}} +\sum\limits_{j\in {\mathcal{K_E}}}{\mv v}_j s_j^{\rm{EH}}, \label{equa:jnl:1}
\end{align}
where ${\mv w}_i\in {\mathbb C}^{M\times 1}$ and ${\mv v}_j\in {\mathbb C}^{M\times 1}$ are the beamforming vectors for ID receiver $i$ and EH receiver $j$, while $s_i^{\rm{ID}}$ and $s_j^{\rm{EH}}$ are the information-bearing signal for ID receiver $i$ and energy-carrying signal  for EH receiver $j$, respectively.  For information signals, Gaussian inputs are assumed, i.e., $s_i^{\rm{ID}}$'s are independent and identically distributed (i.i.d.) CSCG random variables with zero mean and unit variance denoted by $s_i^{\rm{ID}} \sim \mathcal{CN}(0,1), \forall i\in \mathcal{K_{I}}.$ For energy signals, since $s_j^{\rm{EH}}$ carries no information, it can be any arbitrary random signal provided that its power spectral density satisfies certain regulations on microwave radiation. Without loss of generality,  we assume that  $s_j^{\rm{EH}}$'s are independent white sequences from an arbitrary distribution with  $\mathbb{E}\left(|s_j^{\rm{EH}}|^2\right)=1,\forall j\in \mathcal{K_{E}}$. Suppose that the AP has a transmit sum-power constraint $P$; from (\ref{equa:jnl:1}) we thus have
$\mathbb{E}(\mv{x}^H\mv{x}) = \sum\limits_{i\in {\mathcal{K_I}}}\|{\mv w}_i \|^2 +\sum\limits_{j\in {\mathcal{K_E}}}\|{\mv v}_j\|^2 \le P$.

We assume a quasi-static fading environment and denote ${\mv h}_i \in {\mathbb C}^{1\times M}$ and ${\mv g}_j \in {\mathbb C}^{1\times M}$ as the channel vectors from the AP to ID receiver $i$ and EH receiver $j$, respectively, {\color{black}where $\|{\mv h}_i\|^2 = \sigma_{{h,i}}^2$ and $\|{\mv g}_j\|^2 = \sigma_{{g,j}}^2$ with $\sigma_{{g,j}}^2 \gg \sigma_{{h,i}}^2, \forall i\in\mathcal{K_I},j\in\mathcal{K_E}$ (to be consistent with our proposed distance-based information/energy transmission scheme; see Fig. \ref{fig:1}). We also make the following assumptions throughout the paper on the channel independence of different users, which are valid for practical wireless channels in e.g. rich-scattering environments.}
\begin{assumption}[independently distributed user channels]\label{assumption1}
{\color{black}The channel vector ${\mv h}_i$'s and ${\mv g}_j$'s are independently drawn from a set of continuous distribution function $f_{h_i}({\mv h}_i)$'s and $f_{g_j}({\mv g}_j)$'s, respectively, $i\in\mathcal{K_I}, j\in\mathcal{K_E}$. Furthermore, we assume that for any $d\times M$ matrix $\mv{F}$ with $0 < d \le M$, in which the $d$ row vectors constitute any subset of channel vectors from ${\mv h}_i$'s and ${\mv g}_j$'s, it holds with probability one that: i) ${\mathtt{rank}}\left(\mv{F}\right) = d$; and ii) the $d$ (ordered) non-zero singular values of $\mv{F}$, denoted by ${\tau_1}, \cdots,\tau_d$, are strictly decreasing, i.e., ${\tau_1} > \cdots >\tau_d >0$.}
\end{assumption}It is further assumed that the AP  knows perfectly the instantaneous values of $\mv{h}_i$'s and ${\mv g}_j$'s, and each receiver knows its own instantaneous channel.{\footnote{This requires each receiver to perform channel estimation followed by channel feedback to the transmitter, which consumes additional energy. {\color{black}In practice, there exists a design tradeoff at the EH receiver: more accurate channel estimation and feedback may lead to higher harvested energy due to the transmit beamforming gain, but also induce higher energy consumption that can even offset the harvested energy gain  (for detailed discussions on this issue, please refer to \cite{XuZhangICASSP2014}).} For simplicity, we assume in this paper that such energy consumption at EH receivers is negligible compared to their harvested energy.}} The discrete-time baseband signal at the $i$th ID receiver is thus given by
\begin{align}
{y}_i^{\rm{ID}} = {\mv h}_i\mv{x} + z_i,\ \forall i\in \mathcal{K_{I}}, \label{equa:jnl:2}
\end{align}
where $z_i\sim \mathcal{CN}(0,\sigma_i^2)$ is the i.i.d. Gaussian noise at the $i$th ID receiver. With linear transmit precoding, each ID receiver is interfered with by all other non-intended information beams and energy beams. Since energy beams carry no information but instead pseudorandom signals {\color{black}{whose waveforms}} can be assumed to be known at both the AP and each ID receiver prior to data transmission, their resulting interference can be cancelled at each ID receiver if this additional operation is implemented. We thus consider two types of ID receivers, namely Type I and Type II ID receivers, which do not possess and possess the capability of cancelling the interference due to energy signals, respectively. {\color{black}{Furthermore, we assume that the interference precancellation at type II ID receivers is perfect. This assumption is practically valid since each ID receiver knows its own instantaneous channel, and the received energy signals at each ID receiver have the similar dynamic range as the information signals by propagating through the same wireless channel.}} Therefore, for the $i$th ID receiver of Type I or Type II, the corresponding SINR is accordingly expressed as
\begin{align}
\mathtt{SINR}^{(\mathrm{I})}_i&=\frac{|{\mv h}_i{\mv w}_i|^2}{\sum\limits_{k\neq i,k\in\mathcal{K_I}}|{\mv h}_i{\mv w}_k|^2 + \sum\limits_{j\in\mathcal{K_E}}|{\mv h}_i{\mv v}_j|^2 +\sigma_i^2}, \forall i\in \mathcal{K_{I}},  \label{equa:jnl:3}
\end{align}
\begin{align}
\mathtt{SINR}^{(\mathrm{II})}_i=\frac{|{\mv h}_i{\mv w}_i|^2}{\sum\limits_{k\neq i,k\in\mathcal{K_I}}|{\mv h}_i{\mv w}_k|^2  +\sigma_i^2},\ \forall i\in \mathcal{K_{I}}. \label{equa:jnl:4}
\end{align}

On the other hand, for wireless energy transfer, due to the broadcast property of wireless channels, the energy carried by all information and energy beams, i.e., both ${\mv w}_i$'s and ${\mv v}_j$'s, can be harvested at each EH receiver.
As a result, the harvested power for the $j$th EH receiver, denoted by $Q_j$, is proportional to the  total power received \cite{ZhangHo2012}, i.e.,
\begin{align}
Q_j= \zeta\left(\sum\limits_{k\in\mathcal{K_I}}|{\mv g}_j{\mv w}_k|^2  + \sum \limits_{k\in\mathcal{K_E}}|{\mv g}_j{\mv v}_k|^2\right),\ \forall j\in \mathcal{K_{E}},\label{equa:jnl:5}
\end{align}
where $0 < \zeta \le 1$ denotes the energy harvesting efficiency.

Our aim is to maximize the weighted sum-power transferred to all EH receivers subject to individual SINR constraints at different ID receivers, given by $\gamma_i,  i\in\mathcal{K_I}$. Denote $\alpha_j$ as the given energy weight for EH receiver $j$, $\alpha_j\ge 0$, where larger weight value of $\alpha_j$ indicates higher priority of transferring energy to EH receiver $j$ as compared to other EH receivers. Define ${\mv G} = \zeta\sum \limits_{j\in\mathcal{K_E}}\alpha_j {\mv g}_j^H{\mv g}_j$. Then from (\ref{equa:jnl:5}) the  weighted sum-power harvested by all EH receivers can be expressed as
$
\sum \limits_{j\in\mathcal{K_E}} \alpha_j Q_j =  \sum\limits_{i\in\mathcal{K_I}}{\mv w}_i^H{\mv G}{\mv w}_i  + \sum\limits_{j\in\mathcal{K_E}}{\mv v}_j^H{\mv G}{\mv v}_j$. The design problems by assuming that all ID receivers are of either Type I or Type II are thus formulated accordingly as follows.
\begin{align*}
{\mathtt{(P1)}}:\mathop\mathtt{max}_{\{{\mv w}_i\},\{{\mv v}_{j}\}} ~& \sum\limits_{i\in\mathcal{K_I}}{\mv w}_i^H{\mv G}{\mv w}_i  + \sum\limits_{j\in\mathcal{K_E}}{\mv v}_j^H{\mv G}{\mv v}_j  \\
{\mathtt{s.t.}}~~~& \mathtt{SINR}^{(\mathrm{I})}_i \ge \gamma_i,\ \forall i\in\mathcal{K_I} \\
& \sum \limits_{i\in\mathcal{K_I}}\|{{\mv w}_i}\|^2 + \sum \limits_{j\in\mathcal{K_E}}\|{\mv v}_j\|^2 \le P.
\end{align*}
\begin{align*}
{\mathtt{(P2)}}:\mathop\mathtt{max}_{\{{\mv w}_i\},\{{\mv v}_{j}\}} ~&  \sum\limits_{i\in\mathcal{K_I}}{\mv w}_i^H{\mv G}{\mv w}_i  + \sum\limits_{j\in\mathcal{K_E}}{\mv v}_j^H{\mv G}{\mv v}_j \\
{\mathtt{s.t.}}~~~& \mathtt{SINR}^{(\mathrm{II})}_i \ge \gamma_i,\ \forall i\in\mathcal{K_I}\\
&\sum \limits_{i\in\mathcal{K_I}}\|{{\mv w}_i}\|^2 + \sum \limits_{j\in\mathcal{K_E}}\|{\mv v}_j\|^2 \le P.
\end{align*}
Notice that the only difference between ${\mathtt{(P1)}}$  and ${\mathtt{(P2)}}$ lies in the achievable SINR expression for each ID receiver $i\in\mathcal{K_I}$.
Both problems ${\mathtt{(P1)}}$  and ${\mathtt{(P2)}}$ can be shown to maximize a convex quadratic function with $\mv{G}$ being positive semidefinite, i.e., $\mv{G}\succeq \mv{0}$, subject to various quadratic constraints; thus they are both non-convex QCQPs \cite{BoydVandenberghe2004}, for which  the globally optimal solutions are difficult to be obtained efficiently in general.

Prior to solving these two problems, we first have a check on their feasibility, i.e., whether a given set of SINR constraints for ID receivers can be met under the given transmit sum-power constraint $P$. It can be observed from $\mathtt{(P1)}$ and $\mathtt{(P2)}$ that both problems are feasible if and only if their feasibility is guaranteed by ignoring all the EH receivers, i.e.,  setting $\alpha_j=0$ and ${\mv v}_j=\mv{0},\forall j \in\mathcal{K_E}$. It then follows that $\mathtt{SINR}^{(\mathrm{I})}_i = \mathtt{SINR}^{(\mathrm{II})}_i, \forall i\in\mathcal{K_I}$. For convenience, we denote $\mathtt{SINR}^{(\mathrm{I})}_i = \mathtt{SINR}^{(\mathrm{II})}_i \triangleq \mathtt{SINR}_i, \forall i\in\mathcal{K_I}$ in this case.  Thus, the feasibility of both $\mathtt{(P1)}$ and $\mathtt{(P2)}$ can be verified by solving the
following problem:
\begin{align}
\mathop\mathtt{find} ~& {\{{\mv w}_i\}}   \nonumber\\
{\mathtt{s.t.}}~&  \mathtt{SINR}_i \ge \gamma_i,\ \forall i\in\mathcal{K_I} \nonumber\\
~& \sum\limits_{i\in\mathcal{K_I}}\|{{\mv w}_i}\|^2  \le P.\label{equa:jnl:6}
\end{align}
Problem (\ref{equa:jnl:6}) can be solved by {\color{black}standard convex optimization techniques such as the} interior point method via transforming it into a second-order cone program (SOCP) \cite{WieselEldar2006} or by an uplink-downlink duality based fixed-point iteration algorithm \cite{SchubertBoche2004}.

Next, we consider the other extreme case with no ID receivers, i.e., $\mathcal{K_I}=\phi$, where by setting $\mv{w}_i=0,\gamma_i=0,\forall i\in\mathcal{K_I}$, both $\mathtt{(P1)}$ and $\mathtt{(P2)}$ are  reduced to
\begin{align}
\mathop\mathtt{max}_{\{{\mv v}_{j}\}} ~&  \sum\limits_{j\in\mathcal{K_E}}{\mv v}_j^H{\mv G}{\mv v}_j \nonumber \\
{\mathtt{s.t.}}~& \sum \limits_{j\in\mathcal{K_E}}\|{\mv v}_j\|^2 \le P.\label{equa:jnl:7}
\end{align}
Let $\xi_E$ and ${\mv{v}_E}$ be the dominant  eigenvalue and its corresponding eigenvector of ${\mv G}$, respectively. Then it can be easily shown that the optimal value of (\ref{equa:jnl:7}) is $\xi_E P$, which is attained by
setting ${\mv{v}}_j = \sqrt{q_j}{\mv{v}_E},\ \forall j\in\mathcal{K_E},$ for any set of $q_j \ge 0, \forall j\in\mathcal{K_E}$ satisfying $\sum\limits_{j\in\mathcal{K_E}}q_j=P$.
Accordingly, all energy beams are aligned to the same direction as $\mv{v}_E$. Thus, without loss of optimality, we can set $\mv{v}_j = \sqrt{P}\mv{v}_E$ for any $j\in\mathcal{K_E}$ and $\mv{v}_k= \mv{0}, \ \forall k\in\mathcal{K_E}, k\neq j$. For convenience, we refer to the beamformer in the form of $\sqrt{P}\mv{v}_E$ as the optimal energy beamformer (OeBF).

Finally, we consider another special case with all $\gamma_i$'s being sufficiently small (but still non-zero in general), namely the ``OeBF-feasible'' case, in which aligning all information beams to the OeBF is feasible for both ${\mathtt{(P1)}}$ and ${\mathtt{(P2)}}$. In other words, there exists a solution to the feasibility problem (\ref{equa:jnl:6}) given by ${\mv{w}}_i = \sqrt{p_i}{\mv{v}_E},\ \forall i\in\mathcal{K_I}$ with $p_i \ge 0, \forall i\in\mathcal{K_I}$ satisfying $\sum\limits_{i\in\mathcal{K_I}}p_i\le P$, i.e., the following problem has a feasible solution given by $\{p_i\}$.
\begin{align}
\mathop\mathtt{find} ~& \{p_i\}   \nonumber \\
{\mathtt{s.t.}}~& \frac{p_i|{\mv h}_i{\mv v}_E|^2}{\sum\limits_{k\neq i,k\in\mathcal{K_{I}}}p_k|{\mv h}_i{\mv v}_E|^2 + \sigma_i^2} \ge \gamma_i, \forall i\in \mathcal{K_I}\nonumber\\
& \sum_{i\in\mathcal{K_I}}p_i\le P.\label{equation:EOeBFcheck}
\end{align}
In this case, it is easy to verify that the optimal values of both problems $\mathtt{(P1)}$ and $\mathtt{(P2)}$ are $\xi_E P$, which is the same as that of problem (\ref{equa:jnl:7}) and can be attained by ${\mv{w}}_i = \sqrt{\frac{P}{\sum\limits_{k\in\mathcal{K_I}}p_k}p_i}{\mv{v}_E},\ \forall i\in\mathcal{K_I}$ satisfying $\sum\limits_{i\in\mathcal{K_I}}\|{\mv{w}}_i\|^2= P$, and ${\mv{v}}_j = {\mv{0}}, \forall j\in\mathcal{K_E}$, i.e., no dedicated energy beam is needed to achieve the maximum weighted sum-power for EH receivers.  To check whether the OeBF-feasible case occurs or not, we only need to solve the feasibility problem in (\ref{equation:EOeBFcheck}) which is a simple linear program (LP). Therefore, in the rest of this paper, we will mainly focus on the unaddressed non-trivial case so far when $\mathtt{(P1)}$ and $\mathtt{(P2)}$ are both feasible but aligning all information beams to the OeBF is infeasible for both problems, unless otherwise specified.

\section{Optimal Solution via Semidefinite Relaxation}\label{sec:SDP}

In this section, we study the two non-convex QCQPs in $\mathtt{(P1)}$ and $\mathtt{(P2)}$, and derive their optimal solutions via SDR. For non-convex QCQPs, it is known that SDR is an efficient approach to obtain good approximate solutions in general \cite{LuoMaSo2010}. In the following, by applying SDR and exploiting the specific problem structures, the globally optimal solutions for both $\mathtt{(P1)}$ and $\mathtt{(P2)}$ are obtained efficiently.

First, consider problem $\mathtt{(P1)}$ for the case of Type I ID receivers. Define the following matrices: ${\mv W}_i={\mv w}_i{\mv w}_i^H ,\forall i \in \mathcal{K_I}$
and ${\mv W}_E=\sum\limits_{j \in \mathcal{K_E}}\mv v_j\mv v_j^H$. Then, it follows that ${\mathtt{rank}}({\mv W}_i)\le 1,\forall i \in \mathcal{K_I}$
and $\mathtt{rank}({\mv W}_E)\le \mathtt{min}(M,K_E)$. By ignoring the above rank constraints on ${\mv W}_i$'s and ${\mv W}_E$, the SDR of $\mathtt{(P1)}$ is given by
\begin{align*}\mathtt{(SDR1)}:&\\
\mathop{\mathtt{max}}_{\{\mv{W}_i\},\mv{W}_E}
& \sum\limits_{i\in\mathcal{K_I}}\mathtt{tr}(\mv{G}\mv{W}_i)+\mathtt{tr}(\mv{G}\mv{W}_E) \\
\mathtt {s.t.}~~&  \frac{\mathtt{tr}(\mv{h}_i^H\mv{h}_i\mv{W}_i)}{\gamma_i}-\sum\limits_{k\neq i,k\in\mathcal{K_I}}\mathtt{tr}(\mv{h}_i^H\mv{h}_i\mv{W}_k)\\
&~~~~~~~~~~ -\mathtt{tr}(\mv{h}_i^H\mv{h}_i\mv{W}_E)-\sigma_i^2 \geq 0,  \forall i \in \mathcal{K_I}\\ &
\sum\limits_{i\in\mathcal{K_I}}\mathtt{tr}(\mv{W}_i)+\mathtt{tr}(\mv{W}_E)\leq P\\
~& {{\mv W}_i}\succeq {\mv 0}, \forall i\in \mathcal{K_I}, ~~{{\mv W}_E}\succeq {\mv 0}.
\end{align*}
Let the optimal solution of $\mathtt{(SDR1)}$ be $\mv{W}_i^{\star}, \forall i\in\mathcal{K_I}$ and $\mv{W}_E^{\star}$. Then we have the following proposition.
\begin{proposition}\label{proposition:3.1}
{\color{black}Under the condition of independently distributed user channels given in Assumption \ref{assumption1},} the optimal solution of $\mathtt{(SDR1)}$ for the case of Type I ID receivers satisfies: $\mathtt{rank}(\mv{W}_i^{\star}) = 1, \ \forall i\in\mathcal{K_I}$, and $\mv{W}_E^{\star}=\mv{0}$ {\color{black}{with probability one.}}
\end{proposition}
\begin{proof}
See Appendix \ref{appendix:1}.
\end{proof}

From Proposition \ref{proposition:3.1}, it follows that the optimal solution of $\mathtt{(SDR1)}$ satisfies the desired rank constraints, and thus the globally optimal solution of $\mathtt{(P1)}$ can always be obtained by solving $\mathtt{(SDR1)}$. Note that $\mathtt{(SDR1)}$ is a semidefinite program (SDP), which can be efficiently solved by existing software, e.g., $\mathtt{CVX}$ \cite{cvx}. Furthermore, it is observed that the optimal solution satisfies that $\mv{W}_E^\star = \mv{0}$ for $\mathtt{(SDR1)}$ or equivalently $\mv{v}_j = \mv{0}, \ \forall j\in\mathcal{K_E}$ for $\mathtt{(P1)}$, which implies that no dedicated energy beam is needed for achieving the maximum weighted sum harvested power in $\mathtt{(P1)}$. This can be intuitively explained as follows. Since Type I ID receivers cannot cancel the interference from energy beams (if any), employing energy beams will increase the interference power and as a result degrade the SINR at ID receivers. Thus, the optimal transmission strategy is to adjust the weights and power allocation of information beams only to maximize the weighted sum-power transferred to EH receivers.

Next, consider problem $\mathtt{(P2)}$ for the case of Type II ID receivers. Similar to $\mathtt{(P1)}$, the SDR of $\mathtt{(P2)}$ can be expressed as
\begin{align*}
{\mathtt{(SDR2)}}:~\\
\mathop\mathtt{max}_{\{{\mv W}_i\},{\mv W}_E} & \sum\limits_{i\in\mathcal{K_I}}\mathtt{tr}({\mv G}{\mv W}_i) + \mathtt{tr}({\mv G}{\mv W}_E) \\
{\mathtt{s.t.}}~~~~& \frac{\mathtt{tr}({\mv h}_i^H{\mv h}_i{\mv W}_i)}{\gamma_i} - \sum\limits_{k\neq i,k\in\mathcal{K_I}}\mathtt{tr}({\mv h}_i^H{\mv h}_i{\mv W}_k)\\&~~~~~~~~~~~~~~~~~- \sigma_i^2 \geq 0, \forall i \in \mathcal{K_I}\\
~& \sum \limits_{i\in\mathcal{K_I}}\mathtt{tr}({{\mv W}_i}) + \mathtt{tr}({{\mv W}_E}) \le P\\
~& {{\mv W}_i}\succeq {\mv 0},  \forall i \in \mathcal{K_I}, ~~{{\mv W}_E}\succeq {\mv 0}.
\end{align*}
Let the optimal solution of ${\mathtt{(SDR2)}}$ be $\mv{W}_i^*, \forall i\in\mathcal{K_I}$ and $\mv{W}_E^*$. We then have the following proposition.

\begin{proposition}\label{proposition:3.2}
{\color{black}Under the condition of independently distributed user channels given in Assumption \ref{assumption1},} the optimal solution of ${\mathtt{(SDR2)}}$ for the case of Type II ID receivers satisfies:
$\mathtt{rank}({\mv W}_i^*) = 1, \forall i \in \mathcal{K_I}$, $\mathtt{rank}(\mv{W}_E^{*})\le 1$ {\color{black}{with probability one}}; furthermore, it holds that ${\mv W}_{E}^* = q^*\mv{v}_E\mv{v}_E^H$ with $0\le q^* \le P$.
\end{proposition}
\begin{proof}
See Appendix \ref{appendix:2}.
\end{proof}

Based on Proposition {\ref{proposition:3.2}}, we can obtain the globally optimal solution of $\mathtt{(P2)}$ by solving $\mathtt{(SDR2)}$ via $\mathtt{CVX}$. Meanwhile, since ${\mv W}_{E}^* = q^*\mv{v}_E\mv{v}_E^H$, all energy beams should be aligned to $\mv{v}_E$, the same direction as the OeBF. Similar to problem (\ref{equa:jnl:7}), in this case, we can choose to send only one energy beam to minimize the complexity of beamforming implementation at the transmitter as well as the energy signal interference cancellation at all ID receivers by setting
$\mv{v}_j = \sqrt{q^*}\mv{v}_E$ for any $j\in\mathcal{K_E}$ and $\mv{v}_k = \mv{0},  \forall k\in\mathcal{K_E}, k\neq j$.

By comparing the optimal solutions for $\mathtt{(P1)}$ and $\mathtt{(P2)}$, we can see that their main difference lies in {\color{black}whether energy beamforming is employed or not}. Note that the optimal value of $\mathtt{(P2)}$ is in general an upper bound on that of $\mathtt{(P1)}$ since any feasible solution of $\mathtt{(P1)}$ is also feasible for $\mathtt{(P2)}$, but not vice versa. If $q^* = 0$ in Proposition \ref{proposition:3.2},  then the upper bound is tight; however, if $q^* > 0$, then a higher weighted sum harvested power is achievable for EH receivers with Type II ID receivers. Therefore, the benefit of using Type II ID receivers can be realized by employing no more than one energy beam and at the cost of implementing an additional interference cancellation (with {\it a priori} known energy signals) at ID receivers. Nevertheless, it is worth pointing out an interesting case with one single ID receiver, for  which energy beamforming is always not needed, as stated in the following proposition.
\begin{proposition}\label{proposition:KI1}
For the case of Type II ID receivers, if $K_I=1$, then the optimal solution of ${\mathtt{(SDR2)}}$ satisfies that $\mv{W}_E^* = 0$.
\end{proposition}
\begin{proof}
See Appendix \ref{appendix:8}.
\end{proof}
{{\color{black}{
\begin{remark}
It is worth pointing out that in some special channel conditions that do not satisfy Assumption \ref{assumption1} (e.g., in the case of line-of-sight (LOS) user channels some of which happen to be linearly dependent), the tightness of SDRs for $\mathtt{(P1)}$ and $\mathtt{(P2)}$ can still be guaranteed by applying the results in \cite{HuangPalomar2010A}. Consider a separable SDP in the following form:
\begin{align*}
\mathtt{(SSDP)}:\mathop \mathtt{min} _{{\mv X}_1,\ldots,{\mv X}_L}& \sum\limits_{l=1}^L \mathtt{tr}({\mv B}_l {\mv X}_l)\\
\mathtt{s.t.}~& \sum\limits_{l=1}^L \mathtt{tr}({\mv A}_{ml} {\mv X}_l) \unrhd_m b_m, m=1,\ldots,M\\
~& {\mv X}_l \succeq 0,l=1,\ldots,L,
\end{align*}
where ${\mv B}_l$, ${\mv A}_{ml},l=1,\ldots,L,m=1,\ldots,M$ are Hermitian matrices (not necessarily positive semidefinite), $b_m \in {\mathbb R},$ $\unrhd_m \in \{\leq,\geq,=\},m=1,\ldots,M$, and ${\mv X}_l,l=1,\ldots,L$, are Hermitian matrices. Suppose that $\mathtt{(SSDP)}$ is feasible and bounded, and the optimal value is attained. Then $\mathtt{(SSDP)}$ always has an optimal solution $({\mv X}_1^\star,\ldots,{\mv X}_L^\star)$ such that $\sum\limits_{l=1}^L \left({\mathtt{rank}}({\mv X}_l^\star)\right)^2 \leq M$ \cite{HuangPalomar2010A}. By applying this result in our context, it can be verified that there always exists an optimal solution for $\mathtt{(SDR1)}$ satisfying  $\sum\limits_{i\in\mathcal{K_I}}\left( {\mathtt{rank}}({\mv W}_i^\star)\right)^2 + \left({\mathtt{rank}}({\mv W}_E^\star)\right)^2 \leq K_I+1$, and one for $\mathtt{(SDR2)}$  satisfying  $\sum\limits_{i\in\mathcal{K_I}} \left({\mathtt{rank}}({\mv W}_i^*)\right)^2 + \left({\mathtt{rank}}({\mv W}_E^*)\right)^2 \leq K_I+1$. Meanwhile, it can be shown from the SINR constraints that ${\mv W}_i^\star \neq {\mv 0}$ and ${\mv W}_i^* \neq {\mv 0}$, i.e., ${\mathtt{rank}}({\mv W}_i^\star) \ge 1$ and ${\mathtt{rank}}({\mv W}_i^*) \ge 1, \forall i \in \mathcal{K_I}$. Thus, it follows immediately that an optimal solution satisfying $ {\mathtt{rank}}({\mv W}_i^\star) = 1,\ \forall i \in \mathcal{K_I}$, and ${\mathtt{rank}}({\mv W}_E^\star) \le 1$ should exist for $\mathtt{(P1)}$, while one satisfying $ {\mathtt{rank}}({\mv W}_i^*) = 1,\ \forall i \in \mathcal{K_I}$, and ${\mathtt{rank}}({\mv W}_E^*) \le 1$ exists for $\mathtt{(P2)}$. In other words, the SDRs of both $\mathtt{(P1)}$ and $\mathtt{(P2)}$ are still tight even without Assumption 1. Note that the tightness of SDRs in the absence of Assumption \ref{assumption1} can similarly be inferred from \cite[Lemma 1.6]{Bjornson}. However, in general, rank-reduction techniques need to be applied to the higher-rank solutions of SDRs to obtain the rank-one solutions \cite{HuangPalomar2010A}.


It is interesting to compare our work with \cite{HuangPalomar2010A} in more details. Different from \cite{HuangPalomar2010A}, which only shows the existence of rank-one solutions for our problems under general channel conditions, in Propositions \ref{proposition:3.1} and \ref{proposition:3.2} we provide more specific results for the case of independently distributed user channels (cf. Assumption \ref{assumption1}) by applying {\it new proof techniques} (see Appendices \ref{appendix:1} and \ref{appendix:2}). In particular, under Assumption \ref{assumption1}, our results differ from that in \cite{HuangPalomar2010A} in the following two main aspects. First, we show that for energy beamforming, ${\mv W}_E^\star=0$ holds for $\mathtt{(SDR1)}$ and ${\mv W}_E^*=q^*\mv{v}_E\mv{v}_E^H$ holds for $\mathtt{(SDR2)}$ (rather than ${\mathtt{rank}}({\mv W}_E^\star)\le 1$ and ${\mathtt{rank}}({\mv W}_E^*)\le 1$ as inferred from \cite{HuangPalomar2010A}), which provides more insight to the optimal structure of energy beamforming solutions. Second, we show that the optimal information beamforming solutions of $\mathtt{(SDR1)}$ and $\mathtt{(SDR2)}$ are of rank-one with probability one (rather than the existence of rank-one solutions only from \cite{HuangPalomar2010A}); hence, no rank-reduction techniques as in \cite{HuangPalomar2010A} need to be applied in our case.\end{remark}}}}

\section{Alternative Solution Based On Uplink-Downlink Duality}\label{sec:BC-MAC}

In the previous section, we have obtained the globally optimal solutions for our formulated QCQP problems by applying the technique of SDR. In order to provide further insight to the optimal solution structure, in this section we propose an alternative  approach for solving the non-convex problems $\mathtt{(P1)}$ and $\mathtt{(P2)}$ by applying the principle of uplink-downlink duality. It is worth noting that the fundamental reason that some apparently non-convex downlink beamforming problems (e.g., the transmit power minimization problem in \cite{YuLan2007}) can be solved globally optimally via uplink-downlink duality is that they can be recast as certain convex forms (e.g., see the SOCP reformulation for the problem in \cite{YuLan2007}), and thus strong duality holds for these problems. However, for the two non-convex QCQPs in $\mathtt{(P1)}$ and $\mathtt{(P2)}$, we cannot explicitly recast them as convex problems. Nevertheless, our result in Section \ref{sec:SDP} that the SDRs of $\mathtt{(P1)}$ and $\mathtt{(P2)}$ are both tight implies that strong duality may also hold for them. This thus motivates our investigation of a new form of uplink-downlink duality for solving these two problems, as will be shown next. 

\subsection{Algorithm for $\mathtt{(P1)}$ via Uplink-Downlink Duality}

Consider $\mathtt{(P1)}$ for the case with Type I ID receivers at first. According to Proposition \ref{proposition:3.1} {\color{black}and under Assumption \ref{assumption1}}, dedicated energy beamforming is not needed to achieve the optimal solution of $\mathtt{(P1)}$; thus, we can set $\mv{v}_j=\mv{0},\forall j\in\mathcal{K_E}$, and accordingly, $\mathtt{SINR}^{(\mathrm{I})}_i = \mathtt{SINR}^{(\mathrm{II})}_i\triangleq \mathtt{SINR}_i, \forall i\in\mathcal{K_I} $, similar to (\ref{equa:jnl:6}). Hence, $\mathtt{(P1)}$ is reformulated as the following problem.
\begin{align*}
{\mathtt{(P1.1)}}:\mathop\mathtt{max}\limits_{\{{\mv w}_i\}} ~&  \sum\limits_{i\in\mathcal{K_I}}{\mv w}_i^H{\mv G}{\mv w}_i   \\
{\mathtt{s.t.}}~& \mathtt{SINR}_i  \ge \gamma_i, \forall i \in \mathcal{K_I} \\
& \sum \limits_{i\in\mathcal{K_I}}\|{{\mv w}_i}\|^2\le P.
\end{align*}
\begin{figure*}[!t]
\centering
 \epsfxsize=1\linewidth
    \includegraphics[width=15cm]{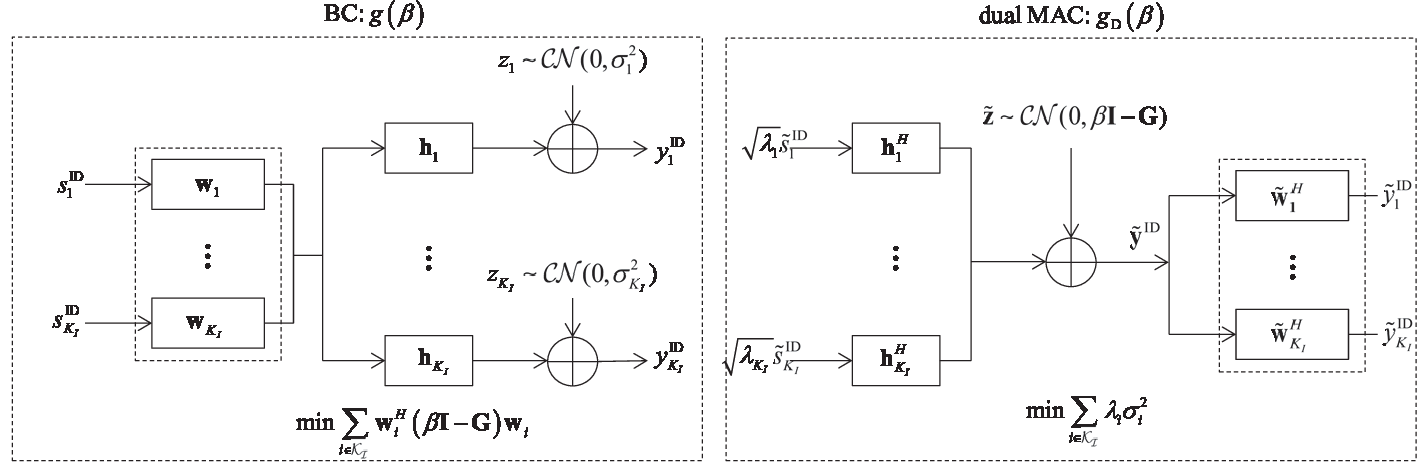}
\caption{Uplink-downlink duality for MISO-BC and SIMO-MAC.} \label{figure:bcmac}
\end{figure*}
By denoting $\beta\ge 0$ as the dual variable associated with the transmit sum-power constraint in ${\mathtt{(P1.1)}}$,{\footnote{\color{black}{Here, we consider the partial Lagrangian formulation of ${\mathtt{(P1.1)}}$ by introducing the dual variable associated with the sum-power constraint only, which is for the convenience of discussing the uplink-downlink duality in different cases    (see Sections \ref{sec:positive} and \ref{sec:negative}). Alternatively, one can derive the same results by directly considering the full Lagrangian formulation of ${\mathtt{(P1.1)}}$ via introducing dual variables for all constraints, which is commonly adopted in existing literature on the uplink-downlink duality related works (see e.g. \cite{YuLan2007,WieselEldar2006}).}}} we can express the dual function of ${\mathtt{(P1.1)}}$ as
\begin{align}
f_1(\beta) \triangleq \mathop\mathtt{max}_{\{{\mv w}_i\}} ~& \sum\limits_{i\in\mathcal{K_I}}{\mv w}_i^H{\mv G}{\mv w}_i - \beta\left(\sum\limits_{i\in\mathcal{K_I}}\|{{\mv w}_i}\|^2 - P\right)  \nonumber \\
{\mathtt{s.t.}}~&  \mathtt{SINR}_i\ge \gamma_i, \forall i \in \mathcal{K_I}.\label{equa:jnl:8}
\end{align}
Accordingly, the dual problem of ${\mathtt{(P1.1)}}$ is defined as
\begin{align*}
{\mathtt{(P1.2)}}:~\mathop \mathtt{min}\limits_{\beta \ge 0}~f_1(\beta).
\end{align*}

Since ${\mathtt{(P1.1)}}$ is known to be non-convex, weak duality holds for ${\mathtt{(P1.1)}}$ and ${\mathtt{(P1.2)}}$ in general, i.e., the optimal value of ${\mathtt{(P1.2)}}$ serves as an upper bound on that of ${\mathtt{(P1.1)}}$. Nevertheless, motivated by the fact that the SDR of ${\mathtt{(P1.1)}}$ is tight, the following proposition establishes that strong duality indeed also holds for ${\mathtt{(P1.1)}}$ and ${\mathtt{(P1.2)}}$.
\begin{proposition}\label{proposition:4.1}
The optimal value of ${\mathtt{(P1.1)}}$ is the same as that of ${\mathtt{(P1.2)}}$.
\end{proposition}
\begin{proof}
See Appendix \ref{appendix:new:D:July17}.
\end{proof}

It follows from Proposition \ref{proposition:4.1} that problems ${\mathtt{(P1.1)}}$ and ${\mathtt{(P1.2)}}$ are equivalent. Thus, we can solve problem ${\mathtt{(P1.1)}}$ by solving problem ${\mathtt{(P1.2)}}$. Specifically, we first solve problem (\ref{equa:jnl:8}) for obtaining $f_1(\beta)$ with any given $\beta\ge 0$, and then search over $\beta \ge 0$ to find the optimal $\beta^\star$ for minimizing $f_1(\beta)$, as will be shown in the following two steps, respectively.

\subsubsection{Obtain $f_1(\beta)$ for given $\beta \ge 0$}\label{sec:all:beta}

Consider problem (\ref{equa:jnl:8}) for obtaining $f_1(\beta)$ with given $\beta\ge 0$. Problem (\ref{equa:jnl:8}) is equivalent to the following problem (by discarding the irrelevant constant term $\beta P$):
\begin{align}
g(\beta) = \mathop\mathtt{min}_{\{{\mv w}_i\}} ~& \sum\limits_{i\in\mathcal{K_I}}{\mv w}_i^H\left(\beta\mv{I}-{\mv G}\right){\mv w}_i   \nonumber \\
{\mathtt{s.t.}}~& \mathtt{SINR}_i  \ge \gamma_i, \forall i\in \mathcal{K_I},\label{equation:13}
\end{align}
where $f_1(\beta)=\beta P-g(\beta)$. We thus focus on solving problem (\ref{equation:13}) to obtain $g(\beta)$ as follows.

Problem (\ref{equation:13}) can be solved by considering a MISO-BC with solely information transmission as shown in the left sub-figure of Fig. \ref{figure:bcmac}, by minimizing the weighted sum transmit power, $\sum\limits_{i\in\mathcal{K_I}}{\mv w}_i^H\left(\beta\mv{I}-{\mv G}\right){\mv w}_i$, subject to a set of individual SINR constraints $\{\gamma_i\}$. For the MISO-BC, its dual single-input multiple-output (SIMO) MAC is shown in the right sub-figure of Fig. \ref{figure:bcmac} by conjugating and transposing the channel vectors, where $K_I$ single-antenna transmitters send independent information to one common receiver with $M$ antennas. At transmitter $i, i\in\mathcal{K_I}$, let $\lambda_i$ be its transmit power, $\tilde{s}^{\rm{ID}}_i$ be a CSCG random variable representing its transmitted information signal, and $\mv{h}_i^{H}$ be its channel vector to the receiver. Then the received signal in the dual SIMO-MAC is expressed as
\begin{align}\label{eqn:simo_mac}
\tilde{\mv{y}}^{\rm{ID}} = \sum\limits_{i\in\mathcal{K_I}}\mv{h}_i^{H}\sqrt{\lambda_i}\tilde{s}^{\rm{ID}}_i + {\tilde{\mv{z}}},
\end{align}
where ${\tilde{\mv{z}}} \sim \mathcal{CN}(\mv{0},\beta\mv{I}-\mv{G})$ denotes the noise vector at the receiver. After applying receive beamforming vector $\tilde{\mv{w}}_i$'s, the SINRs of different users in the dual SIMO-MAC are then given by
\begin{align}\label{eqn:sinr:uplink}
&{\mathtt{SINR}}_i^{\mathrm{MAC}}(\{\tilde{\mv{w}}_i,\lambda_i\}) =\nonumber\\
&\frac{\lambda_i\tilde{\mv{w}}_i^H\mv{h}_i^H\mv{h}_i\tilde{\mv{w}}_i}{\tilde{\mv{w}}_i^H\left(\sum\limits_{k\neq i,k\in\mathcal{K_I}}\lambda_k\mv{h}_k^H\mv{h}_k+\beta\mv{I}-\mv{G}\right)\tilde{\mv{w}}_i},\forall i \in \mathcal{K_I}.
\end{align}
The design objective for the dual SIMO-MAC is to minimize the weighted sum transmit power $\sum\limits_{i\in\mathcal{K_I}}\lambda_i\sigma_i^2$ by jointly optimizing the power allocation $\{\lambda_i\}$ and receive beamforming vectors $\{\tilde{\mv{w}}_i\}$ subject to the same set of SINR constraints $\{\gamma_i\}$ as in the original MISO-BC given by (\ref{equation:13}). We thus formulate the dual uplink problem as
\begin{align}
g_{\mathrm{D}}(\beta)=\mathop{\mathtt{min}}\limits_{\{\tilde{\mv{w}}_i,\lambda_i\}}
& ~\sum\limits_{i\in\mathcal{K_I}}\lambda_i\sigma_i^2\nonumber \\
\mathtt {s.t.} & ~ {\mathtt{SINR}}_i^{\mathrm{MAC}}(\{\tilde{\mv{w}}_i,\lambda_i\}) \geq \gamma_i, \ \forall i \in \mathcal{K_I}\nonumber\\
&~\lambda_i \ge 0,\ \forall i \in \mathcal{K_I}.\label{equa:gd}
\end{align}

Next, we solve the downlink problem (\ref{equation:13}) for any given $\beta\ge 0$ by solving the uplink problem (\ref{equa:gd}) via exploiting the uplink-downlink duality. At first, it is worth noting that if $\beta \ge \xi_E$ (recall that $\xi_E$ is the largest eigenvalue of $\mv{G}$), then $\beta\mv{I}-{\mv G} \succeq \mv{0}$ holds, in which case the downlink weighted sum-power minimization problem in (\ref{equation:13}) can be recast as an equivalent SOCP, or solved based on the dual uplink problem in (\ref{equa:gd}) by the existing algorithm in \cite{YuLan2007}. However, if $\beta < \xi_E$, then it follows that $\beta\mv{I}-{\mv G} \nsucceq \mv{0}$.{\footnote{If $\beta\mv{I}-\mv{G} \nsucceq \mv{0}$, then the receiver noise ${\tilde{\mv{z}}}$ in (\ref{eqn:simo_mac}) for the dual SIMO-MAC cannot be realizable, since the covariance matrix of any physical signal should be positive semidefinite. Thus, in this case, the receiver noise with covariance $\beta\mv{I}-\mv{G} \nsucceq \mv{0}$ is just a mathematical equivalence, and thus needs not be practically realizable.}} In this case, solving problem (\ref{equation:13}) is more involved due to the following reasons. First, the objective function in (\ref{equation:13}) becomes non-convex and as a result, problem (\ref{equation:13}) cannot be recast as a convex (SOCP) problem as in \cite{YuLan2007}. Second, the optimal value of problem (\ref{equation:13}), i.e., $g(\beta)$, may become unbounded from below, e.g.,
when $\beta\mv{I}-{\mv G} \prec \mv{0}$. In the following, we solve problem (\ref{equation:13}) for the two cases of $\beta\ge\xi_E$ and $0\le\beta<\xi_E$, respectively. We first review the uplink-downlink duality based algorithm in \cite{YuLan2007} for solving (\ref{equation:13}) in the case of $\beta \ge \xi_E$, and then extend this algorithm to solve problem (\ref{equation:13}) for the more challenging case of $0 \le \beta < \xi_E$. For convenience, we denote the optimal solutions of problems (\ref{equation:13}) and (\ref{equa:gd}) for any given $\beta\ge 0$ as $\{{{\mv{w}}}_i^{\star}\}$ and $\{\tilde{{\mv{w}}}_i^{\star},\lambda_i^{\star}\}$, respectively.

\paragraph{Solve Problem (\ref{equation:13}) for the case of $\beta \ge \xi_E$ (or $\beta\mv{I}-\mv{G} \succeq \mv{0}$)}\label{sec:positive}

In this case, problems (\ref{equation:13}) and (\ref{equa:gd}) can be shown to be equivalent as in \cite{YuLan2007}. Thus, we can solve the downlink problem (\ref{equation:13}) by first solving the uplink problem (\ref{equa:gd}) and then mapping its solution to that of problem (\ref{equation:13}).

First, consider the uplink problem (\ref{equa:gd}). Since it can be shown that the optimal solution of (\ref{equa:gd}) is always achieved when all the SINR constraints are met with equality \cite{YuLan2007}, it then follows that the optimal uplink transmit power $\{\lambda^{\star}_i\}$ must be a fixed point solution satisfying the following equations \cite{FarrokhiTassiulasLiu98}:
\begin{align}
&\lambda_i^{\star}=\mathtt{m}_i\left(\{\lambda_i^{\star}\}\right)\triangleq\nonumber\\&\mathop\mathtt{min}\limits_{\|\tilde{{\mv{w}}}_i\|=1}\gamma_i\left(\frac{\tilde{\mv{w}}_i^{H}\left(\sum\limits_{k \neq i,k\in\mathcal{K_I}}\lambda_k^{\star}\mv{h}_k^H\mv{h}_k+\beta\mv{I}-\mv{G}\right)\tilde{\mv{w}}_i}{\tilde{\mv{w}}_i^{H}\mv{h}_i^H\mv{h}_i\tilde{\mv{w}}_i}\right), \nonumber\\&~~~~~~~~~~~~~~~~~~~~~~~~~~~~~~~~~~~~~~~~~~~~~~~~~~~~ \forall i \in \mathcal{K_I}.\label{eqn:iterative:2}
\end{align}
As a result, by iterating $\lambda_i^{(n)}=\mathtt{m}_i\left(\{\lambda_i^{(n-1)}\}\right), \forall i \in \mathcal{K_I},$ with $n>0$ being the iteration index, the optimal $\{\lambda^\star_i\}$ for (\ref{equa:gd}) can be obtained. With $\{\lambda^\star_i\}$ at hand, the optimal receive beamforming vector $\{\tilde{\mv{w}}_i^\star\}$ can then be obtained accordingly from (\ref{eqn:iterative:2}) based on the minimum-mean-squared-error (MMSE) principle as
\begin{align}
 \tilde{{\mv{w}}}_i^{\star}~&=~\frac{\left(\sum\limits_{k \neq i,k\in\mathcal{K_I}}\lambda_k^{\star}\mv{h}_k^H\mv{h}_k + \beta\mv{I}-\mv{G}\right)^{\dagger}\mv{h}_i^H}{\left\|\left(\sum\limits_{k \neq i,k\in\mathcal{K_I}}\lambda_k^{\star}\mv{h}_k^H\mv{h}_k + \beta\mv{I}-\mv{G}\right)^{\dagger}\mv{h}_i^H\right\|}, \ \forall i \in \mathcal{K_I}. \label{eqn:iterative:1}
\end{align}

After obtaining the optimal solution of $\{\tilde{\mv{w}}_i^\star,\lambda_i^\star\}$ for the uplink problem (\ref{equa:gd}), we then map the solution to $\{{{\mv{w}}}_i^{\star}\}$ for the downlink problem (\ref{equation:13}). As shown in \cite{YuLan2007}, $\{{{\mv{w}}}_i^{\star}\}$ and $\{\tilde{\mv{w}}_i^\star\}$ are identical up to a certain scaling factor. Using this argument together with the fact that the optimal solution of (\ref{equation:13}) is also attained with all the SINR constraints being tight \cite{YuLan2007} similarly to problem (\ref{equa:gd}), it follows that $\{\mv{w}_i^\star\}$ can be obtained as $\mv{w}_i^\star = \sqrt{p_i^\star} \tilde{\mv{w}}_i^\star, \forall i \in \mathcal{K_I}$, where $\mv{p}^\star = [p_1^\star,\ldots,p_{K_I}^\star]^T$ is given by
\begin{align}
\mv{p}^\star = \bigg(\mv{I}-{\mv{D}}\left(\{\tilde{\mv w}_i^\star,\gamma_i\}\right)\bigg)^{-1}\mv{u}^{\mathrm{BC}}\left(\{\tilde{\mv w}_i^\star, \gamma_i\}\right),\label{solution}
\end{align}
where ${\mv{D}}_{ik}\left(\{\tilde{\mv{w}}_i,\gamma_i\}\right) =\left\{\begin{array}{ll} 0, &
i=k \\ \frac{ \gamma_i|{\mv h}_i\tilde{\mv w}_k|^2}{|{\mv h}_i\tilde{\mv w}_i|^2}, & i\neq k \end{array} \right.$ and $\mv{u}^{\mathrm{BC}}\left(\{\tilde{\mv{w}}_i,\gamma_i\}\right) = \left[\frac{\gamma_1\sigma_1^2}{|{\mv h}_1\tilde{\mv w}_1|^2},\ldots,\frac{\gamma_{K_I}\sigma_{K_I}^2}{|{\mv h}_{K_I}\tilde{\mv w}_{K_I}|^2}\right]^T$.

In summary, Algorithm 1 for solving problem (\ref{equation:13}) for the case of $\beta\ge\xi_E$ is given in Table I.

\begin{table}[htp]
\begin{center}
\caption{Algorithm for Solving Problem (\ref{equation:13}) with Given $\beta\ge\xi_E$}  \vspace{0.1cm}
 \hrule
\vspace{0.1cm} \textbf{Algorithm 1}  \vspace{0.1cm}
\hrule \vspace{0.3cm}
\begin{itemize}
\item[a)] Initialize: $n=0$, and set $\lambda_i^{(0)} \ge 0,\forall i\in\mathcal{K_I}$.
\item[b)] {\bf Repeat:}
    \begin{itemize}
    \item[1)] $n \gets n+1$;
    \item[2)] Update the uplink transmit power as $\lambda_i^{(n)}=\mathtt{m}_i\left(\{\lambda_i^{(n-1)}\}\right), \forall i \in \mathcal{K_I},$ with $\mathtt{m}_i\left(\cdot\right)$ given in (\ref{eqn:iterative:2}).
    \end{itemize}
\item[c)] {\bf Until} $|\lambda_i^{(n)} - \lambda_i^{(n-1)}| \le \epsilon,\forall i \in \mathcal{K_I}$, where $\epsilon$ is the required accuracy.
\item[d)] Set $\lambda_i^{\star}=\lambda_i^{(n)},\forall i \in \mathcal{K_I},$ and compute the uplink receive beamforming vectors $\{\tilde{{\mv{w}}}_i^{\star}\}$ by (\ref{eqn:iterative:1}).
\item[e)] Compute the downlink beamforming vectors as $\mv{w}_i^\star = \sqrt{p_i^\star} \tilde{\mv{w}}_i^\star, \forall i \in \mathcal{K_I}$, where $\{p_i^{\star}\}$ is given by (\ref{solution}).
\end{itemize}
\vspace{0.1cm} \hrule \vspace{0.1cm}\label{algorithm:new0}
\end{center}
\end{table}

\paragraph{Solve Problem (\ref{equation:13}) for the case of $0\le \beta < \xi_E$ (or $\beta\mv{I}-\mv{G} \nsucceq \mv{0}$)}\label{sec:negative}
Given $0\le \beta < \xi_E$, we study further problem (\ref{equation:13}) by considering the two cases where the optimal value of problem (\ref{equation:13}) is bounded from below (i.e., $g(\beta)>-\infty$) and unbounded from below (i.e., $g(\beta)=-\infty$), respectively.

First, we solve problem (\ref{equation:13}) by considering the case of $g(\beta)>-\infty$. In this case, {\color{black}a new form of} uplink-downlink duality is established via the following proposition.

\begin{proposition}\label{theorem:macbc}
If $\beta\mv{I}-\mv{G} \nsucceq \mv{0}$ and $g(\beta)>-\infty$, then the optimal value of (\ref{equation:13}) is equal to that of (\ref{equa:gd}).
\end{proposition}
\begin{proof}
See Appendix \ref{appendix:D:April10}.
\end{proof}

Note that the fundamental reason that Proposition \ref{theorem:macbc} holds is {\color{black}due to the strong duality of problem (\ref{equation:13}) even when $\beta\mv{I}-\mv{G} \nsucceq \mv{0}$, which is a direct consequence of the result that the SDR of problem (\ref{equation:13}) is tight. The use of SDR in establishing the uplink-downlink duality is a new contribution of this paper, which is different from the conventional case of $\beta {\mv I}-\mv G\succeq \mv{0}$, where the uplink-downlink duality has been shown by reformulating (\ref{equation:13}) as an equivalent SOCP \cite{YuLan2007}.} From Proposition \ref{theorem:macbc}, it follows that the downlink problem (\ref{equation:13}) and the uplink problem (\ref{equa:gd}) are still equivalent in this case. Thus, we can solve problem (\ref{equation:13}) by solving problem (\ref{equa:gd}). For problem (\ref{equa:gd}), we obtain the following properties.

\begin{proposition}\label{proposition:July19:2}
If $\beta\mv{I}-\mv{G} \nsucceq \mv{0}$ and $g(\beta)>-\infty$, the optimal solution of problem (\ref{equa:gd}) satisfies that:
\begin{enumerate}
  \item All the SINR constraints are met with equality;
  \item It is true that $\sum\limits_{k \neq i,k\in\mathcal{K_I}}\lambda_k^\star\mv{h}_k^H\mv{h}_k + \beta\mv{I}-\mv{G}\succeq\frac{\lambda_i^\star\mv{h}_i^H\mv{h}_i}{\gamma_i}\succeq \mv{0},\ \forall i \in \mathcal{K_I}$.
\end{enumerate}
\end{proposition}
\begin{proof}
See Appendix \ref{appendix:4}.
\end{proof}

From the first part of Proposition \ref{proposition:July19:2}, it is inferred that the optimal solution of problem (\ref{equa:gd}) must also be a fixed point solution of the equations given in (\ref{eqn:iterative:2}). As a result, the fixed point iteration by $\lambda_i^{(n)}=\mathtt{m}_i\left(\{\lambda_i^{(n-1)}\}\right), \forall i \in \mathcal{K_I}$ given in Algorithm 1 is still applicable for solving the uplink problem (\ref{equa:gd}) in this case. It is worth noting that for the fixed point iteration in this case, at each iteration $n$ we need to ensure that $\sum\limits_{k \neq i,k\in\mathcal{K_I}}\lambda_k^{(n-1)}\mv{h}_k^H\mv{h}_k + \beta\mv{I}-\mv{G}\succeq \mv{0}, \forall i \in \mathcal{K_I}$, since otherwise we will have $\lambda_i^{(n)}=\mathtt{m}_i\left(\{\lambda_i^{(n-1)}\}\right)<0$ if $\sum\limits_{k \neq i,k\in\mathcal{K_I}}\lambda_k^{(n-1)}\mv{h}_k^H\mv{h}_k + \beta\mv{I}-\mv{G}\nsucceq \mv{0}$ for any $i \in \mathcal{K_I}$, which results in an infeasible solution for problem (\ref{equa:gd}). The above requirement can be met by carefully selecting the initial point $\{\lambda_i^{(0)}\}$.

Specifically, we choose $\{\lambda_i^{(0)}\}$ as one feasible solution of problem (\ref{equa:gd}) under the given $\beta<\xi_E$, i.e., $\{\lambda_i^{(0)}\}$ satisfies that $\lambda_i^{(0)}\ge 0$ and ${\mathtt{SINR}}_i^{\mathrm{MAC}}(\{\tilde{\mv{w}}_i,\lambda_i^{(0)}\}) \geq \gamma_i, \ \forall i \in \mathcal{K_I}$ with $\{\tilde{\mv{w}}_i\}$ being any given set of receive beamforming vectors.{\footnote{{\color{black}{Such feasible solution of $\{\tilde{\mv{w}}_i\}$ and $\{\lambda_i^{(0)}\}$ can be constructed as follows. First, set $\{\tilde{\mv{w}}_i\}$ as the normalized vectors of any feasible downlink transmit beamforming vectors for problem (\ref{equa:jnl:6}). Next, under such $\{\tilde{\mv{w}}_i\}$, we can find one set of feasible $\{\lambda_i^{(0)}\}$ by simply solving a linear feasibility problem with the following linear constraints: $\lambda_i^{(0)}\ge 0$ and ${\mathtt{SINR}}_i^{\mathrm{MAC}}(\{\tilde{\mv{w}}_i,\lambda_i^{(0)}\}) \geq \gamma_i, \ \forall i \in \mathcal{K_I}$.}}}} Given such an initial point, the fixed point iteration of $\lambda_i^{(n)}=\mathtt{m}_i\left(\{\lambda_i^{(n-1)}\}\right), \forall i \in \mathcal{K_I}$ will then satisfy the following two properties. First, it yields an element-wise monotonically decreasing sequence of  $\{\lambda_i^{(n)}\}$, i.e., $\lambda_i^{(n)} \le \lambda_i^{(n-1)},  \forall i \in \mathcal{K_I}$. This can be shown based on the fact that $\lambda_i^{(0)}\ge \mathtt{m}_i\left(\{\lambda_i^{(0)}\}\right), \forall i \in \mathcal{K_I}$, given that $\{\lambda_i^{(0)}\}$ is feasible for problem (\ref{equa:gd}). Second, the resulting $\{\lambda_i^{(n)}\}$ is lower bounded by $\{\lambda_i^{\star}\}$, i.e., $\lambda_i^{(n)} \ge \lambda_i^{\star},  \forall i \in \mathcal{K_I}$. This is due to $\lambda_i^{(0)} \ge \lambda_i^\star, \forall i\in\mathcal{K_I}$ together with the fact that given $\lambda_i^{(n-1)} \ge \lambda_i^\star, \forall i\in\mathcal{K_I}$, $\lambda_i^{\star} = \mathtt{m}\left(\{\lambda_i^{\star}\}\right)\le\mathtt{m}\left(\{\lambda_i^{(n-1)}\}\right) = \lambda_i^{(n)},\forall i \in \mathcal{K_I}$, must be true.

By combing the fact that $\lambda_i^{(n)} \ge \lambda_i^{\star},  \forall i \in \mathcal{K_I}$, together with the second part of Proposition \ref{proposition:July19:2}, it then follows that $\sum\limits_{k \neq i,k\in\mathcal{K_I}}\lambda_k^{(n-1)}\mv{h}_k^H\mv{h}_k + \beta\mv{I}-\mv{G}\succeq \mv{0}, \forall i \in \mathcal{K_I}, \forall n$. As a result, the fixed point iteration with the above proposed initial point will converge to a feasible solution for problem (\ref{equa:gd}). Furthermore, in the following proposition, we show that this converged feasible solution is indeed optimal.

\begin{proposition}\label{theorem:converge}
If $\beta\mv{I}-\mv{G} \nsucceq \mv{0}$ and $g(\beta)>-\infty$, then the fixed point iteration converges to the optimal solution $\{\lambda_i^{\star}\}$ for problem (\ref{equa:gd}).
\end{proposition}
\begin{proof}
See Appendix \ref{appendix:7}.
\end{proof}

With the optimal $\{\lambda_i^{\star}\}$ at hand, $\{\tilde{\mv{w}}_i^{\star}\}$ can be obtained from (\ref{eqn:iterative:1}). Thus, we have solved the uplink problem (\ref{equa:gd}) in this case.

We then map $\{\tilde{\mv{w}}_i^{\star},\lambda^\star_i\}$ for the uplink problem (\ref{equa:gd}) to $\{{\mv{w}}_i^{\star}\}$ for the downlink problem (\ref{equation:13}). Similar to the case of $\beta\mv{I}-\mv{G} \succeq \mv{0}$,  $\{{\mv{w}}_i^{\star}\}$ can be obtained as $\mv{w}_i^\star = \sqrt{p_i^\star} \tilde{\mv{w}}_i^\star, \forall i \in \mathcal{K_I}$, with $\mv{p}^\star = [p_1^\star,\ldots,p_{K_I}^\star]^T$ given by (\ref{solution}). Therefore, we have solved the downlink problem (\ref{equation:13}) when $g(\beta)$ is bounded from below.

Next, we consider the case of $g(\beta)=-\infty$. In this case, problems (\ref{equation:13}) and (\ref{equa:gd}) are no more equivalent, since it is evident that the optimal value of problem (\ref{equa:gd}) should be no smaller than zero, i.e., $g_{\rm D}(\beta) \ge 0$, and thus $g(\beta) < g_{\rm D}(\beta)$ must be true. However, we can still apply the fixed point iteration of $\lambda_i^{(n)}=\mathtt{m}_i\left(\{\lambda_i^{(n-1)}\}\right), \forall i \in \mathcal{K_I}$ together with an initial feasible point $\{\lambda_i^{(0)}\}$ to solve problem (\ref{equa:gd}), provided that we check the unboundedness by 	
examining the positive semi-definiteness of the matrix $\sum\limits_{k \neq i,k\in\mathcal{K_I}}\lambda_k^{(n)}\mv{h}_k^H\mv{h}_k + \beta\mv{I}-\mv{G}, \forall i\in\mathcal{K_I}$. More specifically, we have the following proposition.
\begin{proposition}\label{proposition:4.5}
If $\beta\mv{I}-\mv{G} \nsucceq \mv{0}$ and $g(\beta)=-\infty$, then the fixed point iteration always converges to a solution with $\sum\limits_{k \neq i,k\in\mathcal{K_I}}\lambda_k^{(n)}\mv{h}_k^H\mv{h}_k + \beta\mv{I}-\mv{G}\nsucceq \mv{0}$ for some $i\in\mathcal{K_I}$ for problem (\ref{equa:gd}).
\end{proposition}
\begin{proof}
See Appendix \ref{appendix:6:add}.
\end{proof}
Proposition \ref{proposition:4.5} thus provides an efficient way to check the unboundedness of $g(\beta)$ when $g(\beta)=-\infty$.

It is interesting to make a comparison between the two cases of $g(\beta)>-\infty$ and $g(\beta)=-\infty$ in solving the uplink problem (\ref{equa:gd}) by the fixed point iteration. With an initial feasible point for both cases, an element-wise monotonically decreasing sequence of  $\{\lambda_i^{(n)}\}$ is obtained by the fixed point iteration. However, $\{\lambda_i^{(n)}\}$ is lower bounded by $\{\lambda_i^{\star}\}$ in the former case, while it is unbounded from below in the latter case. Therefore, the same fixed point iteration will lead to different converged solutions for the two cases.

To summarize, an algorithm for solving problem (\ref{equation:13}) with given $0\le\beta<\xi_E$ is provided in Table II as Algorithm 2. Note that Algorithm 2 differs from Algorithm 1 in two main aspects: First, in step a), the initial point $\{\lambda_i^{(0)}\}$ should be set as a feasible solution for problem (\ref{equa:gd}); and second, step b-2) is added to check the unboundedness for $g(\beta)$.

\begin{table}[htp]
\begin{center}
\caption{ Algorithm for Solving Problem (\ref{equation:13}) with given $0\le\beta<\xi_E$}  \vspace{0.1cm}
 \hrule
\vspace{0.1cm} \textbf{Algorithm 2}  \vspace{0.1cm}
\hrule \vspace{0.1cm}
\begin{itemize}
\item[a)] Initialize: $n=0$, and set $\lambda_i^{(0)} \ge 0,\forall i\in\mathcal{K_I},$ as a feasible solution of problem (\ref{equa:gd}).
\item[b)] {\bf Repeat:}
    \begin{itemize}
    \item[1)] $n \gets n+1$;
    \item[2)] Check whether there exists an $i\in\mathcal{K_I}$ such that $\sum\limits_{k \neq i,k\in\mathcal{K_I}}\lambda_k^{(n-1)}\mv{h}_k^H\mv{h}_k + \beta\mv{I}-\mv{G}\nsucceq \mv{0}$. If yes, set $g(\beta)=-\infty$, and exit the algorithm. Otherwise, continue;
    \item[3)] Update the uplink transmit power as $\lambda_i^{(n)}=\mathtt{m}_i\left(\{\lambda_i^{(n-1)}\}\right), \forall i \in \mathcal{K_I},$ with $\mathtt{m}_i\left(\cdot\right)$ given in (\ref{eqn:iterative:2}).
    \end{itemize}
\item[c)] {\bf Until} $|\lambda_i^{(n)} - \lambda_i^{(n-1)}| \le \epsilon,\forall i \in \mathcal{K_I}$, where $\epsilon$ is the required accuracy.
\item[d)] Set $\lambda_i^{\star}=\lambda_i^{(n)},\forall i \in \mathcal{K_I},$ and compute the uplink receive beamforming vectors $\{\tilde{{\mv{w}}}_i^{\star}\}$ by (\ref{eqn:iterative:1}).
\item[e)] Compute the downlink beamforming vectors as $\mv{w}_i^\star = \sqrt{p_i^\star} \tilde{\mv{w}}_i^\star, \forall i \in \mathcal{K_I}$, where $\{p_i^{\star}\}$ is given by (\ref{solution}).
\end{itemize}
\vspace{0.1cm} \hrule \vspace{0.1cm}\label{algorithm:2}
\end{center}
\end{table}

\subsubsection{Minimize $f_1(\beta)$ over $\beta\ge 0$}\label{sec:search}

By combing the solutions to problem (\ref{equation:13}) for the two cases of $\beta \ge \xi_E$ and $0 \le \beta < \xi_E$, we obtain $g(\beta)$ and thus $f_1(\beta)=\beta P-g(\beta),\forall \beta \ge 0$. We are then ready to solve $\mathtt{(P1.1)}$ by finding the optimal $\beta^\star \ge 0$ to minimize $f_1(\beta)$. It is easy to show that $f_1(\beta)$ is a convex function, for which the subgradient at given $\beta\ge 0$ is $\upsilon(\beta)\triangleq P-\sum\limits_{i\in\mathcal{K_I}}{\mv w}_i^{\star H} {\mv w}^{\star}_i$ if $f_1(\beta)<\infty$ in the case of $g(\beta)> -\infty$. On the other hand, if $f_1(\beta)=\infty$ in the case of $g(\beta)=-\infty$, then it is evident that $\beta^\star > \beta$. By applying the above two results, we can thus use the simple bisection method to obtain the optimal $\beta^\star$ to minimize $f_1(\beta)$. As a result, the optimal beamforming solution $\{{\mv w}^{\star}_i\}$ in (\ref{equation:13}) corresponding to $\beta^\star$ becomes the optimal solution of $\mathtt{(P1.1)}$.

\subsection{Algorithm for $\mathtt{(P2)}$ via Uplink-Downlink Duality}

Next, consider $\mathtt{(P2)}$ for the case with Type II ID receivers. As shown in Proposition \ref{proposition:3.2}, employing only one energy beam aligning to the OeBF is optimal for $\mathtt{(P2)}$. Hence, we can replace ${\mv v}_{1}, \ldots, {\mv v}_{K_E}$ by one common energy beam $\mv w_E = \sqrt{q}\mv{v}_E, q\ge 0$ without loss of optimality for $\mathtt{(P2)}$. By noting that $\mv w_E^H{\mv G}{\mv w}_E = q\xi_E$ and $\|{\mv w}_E\|^2 = q$, we accordingly reformulate $\mathtt{(P2)}$ as
\begin{align*}
{\mathtt{(P2.1)}}:\mathop\mathtt{max}_{\{{\mv w}_i\},q\ge 0} ~&  \sum\limits_{i\in\mathcal{K_I}}{\mv w}_i^H{\mv G}{\mv w}_i  + q\xi_E\\
{\mathtt{s.t.}}~& \mathtt{SINR}_i \ge \gamma_i, \forall i \in \mathcal{K_I} \\
& \sum \limits_{i\in\mathcal{K_I}}\|{{\mv w}_i}\|^2 + q \le P.
\end{align*}
Similar to ${\mathtt{(P1.1)}}$, we introduce the dual function of ${\mathtt{(P2.1)}}$ as

\begin{align}
f_2(\beta) \triangleq &\nonumber \\\mathop\mathtt{max}_{\{{\mv w}_i\},q\ge 0} &\sum\limits_{i\in\mathcal{K_I}}{\mv w}_i^H{\mv G}{\mv w}_i  +  q\xi_E - \beta\left(\sum\limits_{i\in\mathcal{K_I}}\|{{\mv w}_i}\|^2  + q - P\right) \nonumber \\
{\mathtt{s.t.}}~& \mathtt{SINR}_i \ge \gamma_i, \forall i \in \mathcal{K_I},\label{equa:jnl:9}
\end{align}
where $\beta\ge0$ is the dual variable associated with the transmit sum-power constraint in ${\mathtt{(P2.1)}}$. Accordingly, the dual problem can be defined as
\begin{align*}
\mathtt{(P2.2)}:\mathop\mathtt{min}_{\beta \ge 0}f_2(\beta).
\end{align*}
We then have the following proposition.
\begin{proposition}\label{proposition:4.4}
Strong duality holds between ${\mathtt{(P2.1)}}$ and ${\mathtt{(P2.2)}}$.
\end{proposition}
\begin{proof}
The proof is similar to that of Proposition \ref{proposition:4.1} and thus is omitted for brevity.
\end{proof}

Given the above proposition, we can solve ${\mathtt{(P2.1)}}$ by solving ${\mathtt{(P2.2)}}$, i.e., first solving problem (\ref{equa:jnl:9}) to obtain $f_2(\beta)$ for given $\beta\ge 0$ and then searching the optimal $\beta \ge 0$, denoted by $\beta^*$, to minimize $f_2(\beta)$.

First, consider problem (\ref{equa:jnl:9}) for given $\beta\ge 0$. We then have the following proposition.
\begin{proposition}\label{proposition:4.9}
In order for $f_2(\beta)$ to be bounded from above, it must hold that $\beta \ge \xi_E$.
\end{proposition}
\begin{proof}
Suppose that $\beta < \xi_E$. In this case, it is easy to verify that the objective value of problem  (\ref{equa:jnl:9}) goes to infinity as $q \to \infty$, i.e.,  $f_2(\beta)$ is unbounded from above. Therefore, $\beta < \xi_E$ cannot hold in order for $f_2(\beta)$ to be bounded from above. This proposition is thus proved.
\end{proof}

Proposition \ref{proposition:4.9} specifies that $f_2(\beta)=\infty$ if $\beta<\xi_E$, and thus $\beta^*\ge \xi_E$ must hold for ${\mathtt{(P2.2)}}$. Notice that this result is different from the case of ${\mathtt{(P1.2)}}$  where $f_1(\beta)$ can be bounded from above even when $\beta<\xi_E$ and thus $\beta^{\star}<\xi_E$ may hold for ${\mathtt{(P1.2)}}$.

Given $\beta\ge\xi_E$, problem (\ref{equa:jnl:9}) is solved as follows. First, we express $f_2(\beta) = \beta P + h(\beta)-g(\beta)$ with $g(\beta)$ defined in (\ref{equation:13}), and $h(\beta)$ given by
\begin{align}
h(\beta) = \mathop\mathtt{max}_{q\ge 0} q(\xi_E-\beta).\label{equation:13:4.1}
\end{align}
Accordingly, problem  (\ref{equa:jnl:9}) can be decomposed into two subproblems (by discarding the irrelevant term $\beta P$), which are problem (\ref{equation:13}) for obtaining $g(\beta)$ and problem (\ref{equation:13:4.1}) for obtaining $h(\beta)$, respectively.

For problem (\ref{equation:13}) with $\beta \ge \xi_E$, Algorithm 1 in Table I directly applies to obtain the optimal solution of $\{{\mv w}_i^{\star}\}$. For problem (\ref{equation:13:4.1}) with $\beta \ge \xi_E$, it is easily verified that one solution is given by $q^{\star} = 0.$\footnote{Note that if $\beta=\xi_E$, then the optimal solution of $q$ is non-unique and can take any non-negative value in problem (\ref{equation:13:4.1}). For convenience, we let $q^\star = 0$ in this case.} With both $g(\beta)$ and $h(\beta)=0$ at hand, we can obtain $f_2(\beta)$ for the case of $\beta \ge \xi_E$. Then, we solve problem $\mathtt{(P2.2)}$ by finding the optimal $\beta^*$ to minimize $f_2(\beta)$. Since $f_2(\beta)$ is a convex function, we can apply the bisection method to minimize it over $\beta\ge\xi_E$, given that the subgradient of $f_2(\beta)$ at given $\beta\ge \xi_E$ can be shown to be $P-\sum\limits_{i\in\mathcal{K_I}}{\mv w}_i^{\star H} {\mv w}^{\star}_i-q^{\star}=P-\sum\limits_{i\in\mathcal{K_I}}{\mv w}_i^{\star H} {\mv w}^{\star}_i\triangleq\upsilon(\beta)$. Therefore, the optimal solution of $\mathtt{(P2.2)}$ can be obtained as $\beta^*$. Then, the corresponding solution $\{{\mv w}_i^\star\}$ for problem (\ref{equation:13}) becomes the optimal solution for problem $\mathtt{(P2.1)}$, denoted by $\{{\mv w}_i^*\}$. It should be pointed out that the optimal solution of $q$ in $\mathtt{(P2.1)}$, denoted by $q^{*}$, cannot be directly obtained as $q^* = 0$ if $\beta^*=\xi_E$; but instead $q^{*}$ should be obtained from the complementary slackness condition \cite{BoydVandenberghe2004}  $\beta^*\left(\sum\limits_{i\in\mathcal{K_I}}{\mv w}_i^{*H}{\mv w}_i^*  + q^{*} - P\right) = 0$ as $
q^{*}=P - \sum\limits_{i\in\mathcal{K_I}}{\mv w}_i^{*H}{\mv w}_i^*.
$
Therefore, problem $\mathtt{(P2.1)}$ is solved.

\subsection{Solution Comparison with Type I versus Type II ID Receivers}

\begin{figure}
\centering
 \epsfxsize=1\linewidth
    \includegraphics[width=8cm]{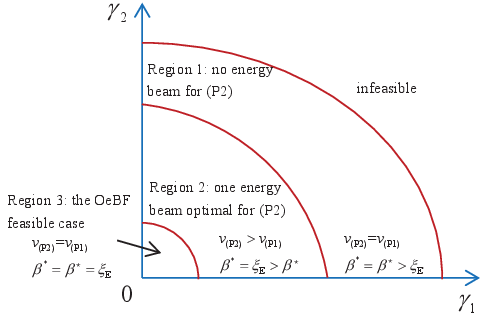}
\caption{Illustration of the optimal solution for $(\mathtt{P1})$ and $(\mathtt{P2})$.} \label{fig:ill2}
\end{figure}

Finally, we compare the optimal solutions for problems $(\mathtt{P1})$ or $(\mathtt{P1.1})$ with Type I ID receivers versus $(\mathtt{P2})$ or $(\mathtt{P2.1})$ with Type II ID receivers. We denote the optimal values of $(\mathtt{P1})$ and $(\mathtt{P2})$ as $v_{\mathtt{(P1)}}$ and $v_{\mathtt{(P2)}}$, respectively. From (\ref{equa:jnl:8}), (\ref{equation:13}), (\ref{equa:jnl:9}) and (\ref{equation:13:4.1}), and by noting that $h(\beta^*)=0$, it follows that $v_{\mathtt{(P1)}} = f_1(\beta^\star) = \beta^\star P-g(\beta^\star)$ and $v_{\mathtt{(P2)}}= f_2(\beta^*)= \beta^* P-g(\beta^*)$, where $\beta^\star$ and $\beta^*$ are the optimal dual solutions for $(\mathtt{P1})$ and $(\mathtt{P2})$, respectively. By observing that $\beta^*\ge \xi_E$ in $(\mathtt{P2})$ while both $\beta^\star\ge \xi_E$ and  $\beta^\star< \xi_E$ can occur in $(\mathtt{P1})$, we compare their optimal values based on $\beta^\star$ and $\beta^*$ over the region of all feasible SINR targets for three cases, where each case corresponds to one subregion as shown in Fig. \ref{fig:ill2} for the case of two ID receivers. In the first subregion with $\beta^{*}=\beta^\star>\xi_E$, it follows that $v_{\mathtt{(P1)}} = v_{\mathtt{(P2)}}$ and $q^{*} = 0$ in $(\mathtt{P2})$, indicated as Region 1 in Fig. \ref{fig:ill2}. In this case with sufficiently large SINR constraint values, the transmit power should be all used for information beams to ensure that the SINR constraints at ID receivers are all met, and no dedicated energy beam is needed for the optimal solutions of both $(\mathtt{P1})$ and $(\mathtt{P2})$. In the second subregion with $\beta^*=\xi_E >\beta^\star$, it follows that $v_{\mathtt{(P2)}} > v_{\mathtt{(P1)}}$ and $q^* > 0$ in $(\mathtt{P2})$, indicated as  Region 2 in Fig. \ref{fig:ill2}. In this case with moderate SINR constraint values, employing one energy beam is beneficial for Type II ID receivers as compared to no energy beam for Type I ID receivers. In the third subregion with $\beta^*=\beta^\star=\xi_E$, it follows that $v_{\mathtt{(P2)}}=v_{\mathtt{(P1)}}$ and $q^{*} \ge 0$ in $(\mathtt{P2})$, shown as Region 3 in Fig. \ref{fig:ill2}, which is the OeBF-feasible case given in Section \ref{sec:system} for sufficiently small SINR constraint values.

\section{Simulation Results}\label{sec:simulation}

In this section, we provide numerical examples to validate our results. We assume that the signal attenuation from the AP to all EH receivers is 30 dB corresponding to an equal distance of 1 meter, and that to all ID receivers is 70 dB at an equal distance of 20 meters. The channel vector $\mv{g}_j$'s and $\mv{h}_i$'s are randomly generated from i.i.d. Rayleigh fading {\color{black}(thus, satisfying Assumption \ref{assumption1})} with the  average channel powers set according to the above average attenuation values. We set $P=1$ Watt(W)  or 30 dBm, $\zeta=50\%$, $\sigma_i^2 = -50$ dBm, and $\gamma_i = \gamma, \forall i\in\mathcal{K_I}$. We also set $\alpha_j = \frac{1}{K_E}, \forall j\in\mathcal{K_E}$; thus the average harvested power of all EH receivers is considered.
\subsection{Performance Comparison of Type I versus Type II ID Receivers}

\begin{figure}
\centering
 \epsfxsize=1\linewidth
    \includegraphics[width=8cm]{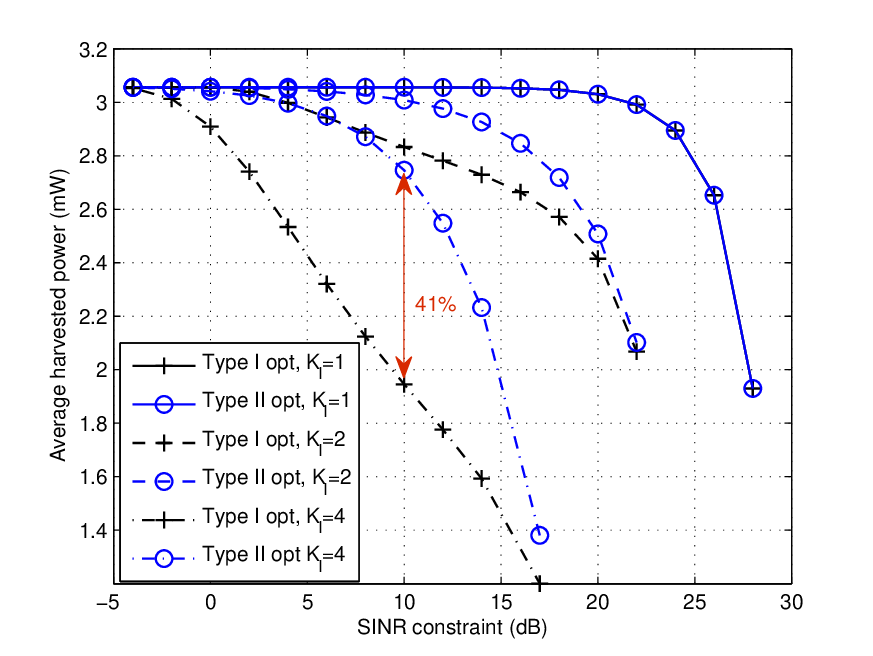}
\caption{Average harvested power versus SINR constraint with optimal beamforming designs.} \label{fig:5}
\end{figure}
Fig. \ref{fig:5} compares the average harvested power obtained by solving $\mathtt{(P1)}$ for Type I ID receivers  and that by $\mathtt{(P2)}$ for Type II ID receivers versus different SINR constraint values of $\gamma$ with fixed $M=4$ and $K_E=2$ and over 200 random channel realizations. It is observed that Type I and Type II ID receivers  have the same performance  when $K_I=1$, which is consistent with Proposition \ref{proposition:KI1}. With $K_I=2$ or $4$, it is observed that Type I and Type II ID receivers have similar performance when $\gamma$ is either large or small, while the latter  outperforms the former notably for moderate values of $\gamma$. The reasons can be explained by referring to Fig. \ref{fig:ill2} as follows. When $\gamma$ is sufficiently small, the OeBF-feasible case shown as Region 3 in Fig. \ref{fig:ill2} holds, where aligning all information beams in the direction of the OeBF is not only feasible but also optimal for both $\mathtt{(P1)}$ and $\mathtt{(P2)}$; thus, the same performance for both types of ID receivers is observed in Fig. \ref{fig:5}. On the other hand, when $\gamma$ is sufficiently large, this case corresponds to Region 1 in Fig. \ref{fig:ill2}, in which it is optimal to allocate all transmit power to information beams to ensure that the SINR constraints at ID receivers are all met; as a result, transmit power allocated to energy beams is zero for both types of ID receivers, and thus their performances are also identical. At last, for the case of moderate values of  $\gamma$ which corresponds to Region 2 in Fig. \ref{fig:ill2}, the considerable performance gain by Type II over Type I ID receivers is due to the use of one dedicated energy beam. For example, {\color{black}{under this particular channel setup,}} as shown in Fig. \ref{fig:5}, a 41\% average harvested power gain is achieved for EH receivers with Type II ID receivers as compared to Type I ID receivers when $\gamma=10$ dB and $K_I=4$, thanks to the cancellation of (known) energy signals at ID receivers.

\begin{figure}
\centering
 \epsfxsize=1\linewidth
    \includegraphics[width=8cm]{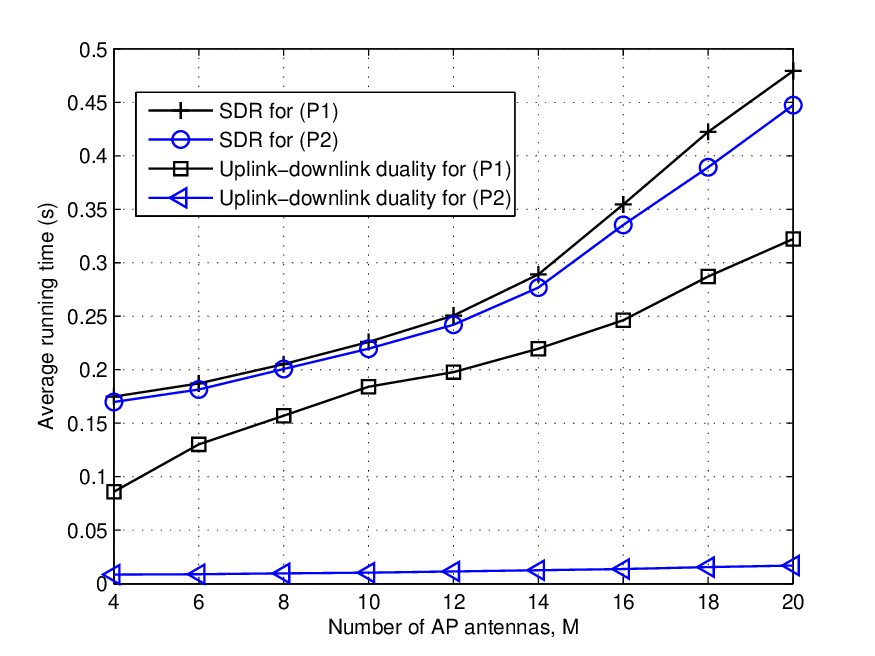}
\caption{Run-time comparison of SDR and uplink-downlink duality based algorithms.} \label{fig:Complexity}
\end{figure}

\subsection{Complexity Comparison of SDR and Uplink-Downlink Duality Based Algorithms}

Next, we compare the complexity of the SDR and uplink-downlink duality based algorithms for solving $\mathtt{(P1)}$ and $\mathtt{(P2)}$, by evaluating their average running times.{\footnote{It has been shown in \cite{HuangPalomar2} that the SDRs of $\mathtt{(P1)}$ and $\mathtt{(P2)}$ can be solved with a worst-case complexity of $\mathcal{O}\left( \left(K_I^3 M^{3.5} + K_I^4\right) \log(1/\varepsilon) \right)$, given a solution accuracy $\varepsilon > 0$. However, we cannot obtain the analytic complexity orders of the two uplink-downlink duality based algorithms, since they depend on both the inner fixed-point iteration (say, parameter $\epsilon$ in Algorithms 1 and 2) and the outer bisection iteration. Therefore, it is difficult to provide a rigorous analytic complexity comparison for the two approaches.}} We conduct the simulations by using Matlab on a computer equipped with an Intel Core i5-2500 @3.3GHz processor and 8GB of RAM memory.

Fig. \ref{fig:Complexity} shows the average running times of different algorithms versus the number of transmit antennas at the AP $M$ with fixed $K_I=4$, $K_E=2$, and $\gamma = 10$ dB. It is observed that for solving either $\mathtt{(P1)}$ or $\mathtt{(P2)}$, the SDR based algorithm has a longer running time than the uplink-downlink duality based algorithm for a given $M$. This is due to the fact that the SDR is performed over matrices with much higher number of unknowns than that of the uplink-downlink duality based algorithm involving beamforming vectors only. It is also observed that the uplink-downlink duality based algorithm for solving $\mathtt{(P1)}$ consumes much longer running time than that for $\mathtt{(P2)}$. This is because in the former case the algorithm needs to check the positive semidefiniteness of $\sum\limits_{k \neq i,k\in\mathcal{K_I}}\lambda_k^{(n-1)}\mv{h}_k^H\mv{h}_k + \beta\mv{I}-\mv{G}, \forall {i\in\mathcal{K_I}}$, in each iteration of $n$ when implementing Algorithm 2 (cf. step b-2) in Algorithm 2), which takes additional running time.

\subsection{Performance Comparison of Optimal versus Suboptimal Designs}

Finally, we compare the performances of our proposed optimal joint information/energy beamforming designs with two suboptimal designs for Type I and Type II ID receivers, respectively, which are described as follows.
\subsubsection{Separate information/energy beamforming design with Type I ID receivers} In this scheme, the information beams are first designed to minimize the required transmit sum-power for satisfying the SINR constraints at all ID receivers, while one energy beam is then added to maximize the weighted sum-power harvested by the EH receivers with the remaining power subject to the constraint of no interference to all ID receivers (since Type I ID receivers are considered here and thus any interference from energy signals cannot be cancelled at ID receivers). Notice that this scheme is applicable only for the case of $K_I \le M-1$. First, the information beams are obtained by solving the following problem:
\begin{align}
\{{\mv w}_i^{\rm{min}}\} =
\mathtt{arg}\mathop\mathtt{min}_{\{{\mv w}_i\}} ~& \sum \limits_{i\in\mathcal{K_I}}\|{{\mv w}_i}\|^2\nonumber \\
\mathtt{s.t.}~&  \mathtt{SINR}_i \ge {\gamma_i}, \forall i \in \mathcal{K_I},\label{eqn:30}
\end{align}
which can be solved by conventional methods such as the fixed-point iteration based on the uplink-downlink duality similar to Algorithm 1. After obtaining ${\mv w}_i = {\mv w}_i^{\rm{min}}, \forall i\in\mathcal{K_I}$, the energy beam ${\mv w}_E$ is then optimized over the null space of $\mv{H}=\left[\mv{h}_1^T~\cdots ~\mv{h}_{K_I}^T\right]^T$, which can be obtained by solving the following problem:
        \begin{align}
        \mathop\mathtt{max}_{{\mv w}_{E}} ~&  {\mv w}_E^H{\mv G}{\mv w}_E \nonumber \\
        {\mathtt{s.t.}}~& \mv{H}{\mv w}_E = \mv{0},\nonumber\\
        &\|{\mv w}_E\|^2 \le P-\sum\limits_{i\in\mathcal{K_I}}\|{\mv w}_i^{\rm{min}}\|^2.\label{equa:jnl:sep}
        \end{align}
Let the singular value decomposition (SVD) of $\mv{H}$ be given by $\mv{H} = \mv{U}\mv{\Lambda}\left[{\mv{V}}~\bar{\mv{V}}\right]^H$, where $\bar{\mv{V}}\in\mathbb{C}^{M\times(M-K_I)}$ consists of the vectors corresponding to zero singular values of $\mv{H}$ and thus spans the null space of $\mv{H}$. Then the optimal solution of (\ref{equa:jnl:sep}) can be obtained as ${\mv w}_E=\sqrt{P-\sum\limits_{i\in\mathcal{K_I}}\|{\mv w}_i^{\rm{min}}\|^2}\bar{\mv{V}}{\mv v}'_E$ with ${\mv v}'_E$ being the dominant eigenvector of $\bar{\mv{V}}^{H}\mv{G}\bar{\mv{V}}$.
\subsubsection{Separate information/energy beamforming design with Type II ID receivers} In this scheme, we first set the information beams to satisfy the SINR constraints at all ID receivers with the minimum  transmit sum-power as ${\mv w}_i = {\mv w}_i^{\rm{min}}, \forall i\in\mathcal{K_I}$, given in (\ref{eqn:30}). Then, we allocate the remaining power to the energy beam aligning to the OeBF to maximize the weighted sum-power transferred to EH receivers (since Type II ID receivers are considered in this case, which can cancel the interference due to energy signals), which is given by ${\mv w}_E=\sqrt{P-\sum\limits_{i\in\mathcal{K_I}}\|{\mv w}_i^{\rm{min}}\|^2}{\mv v}_E$.

\begin{figure}
\centering
 \epsfxsize=1\linewidth
    \includegraphics[width=8cm]{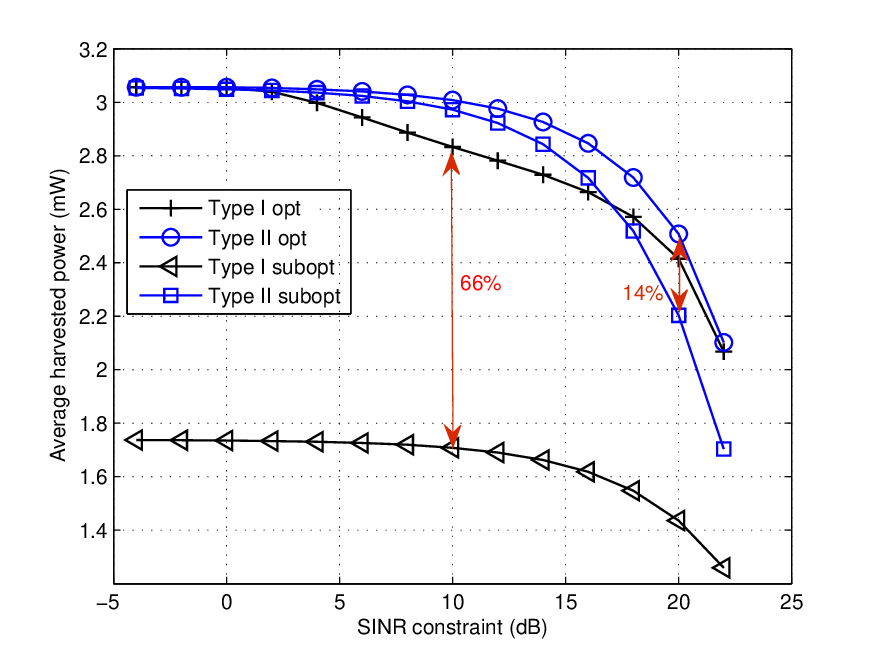}
\caption{Performance comparison of optimal versus suboptimal beamforming designs.} \label{fig:6}
\end{figure}

Fig. \ref{fig:6} compares the average harvested power over the SINR constraint for both optimal and suboptimal designs for the two types of ID receivers, where $M=4, K_E=2$  and $K_I=2$. With Type I ID receivers, it is observed that the separate information and energy beamforming design approach performs severely worse than the optimal joint design. In contrast, with Type II ID receivers, it is observed that the separate design obtains comparable performance to the joint design, especially when $\gamma$ is small. From this result, it is inferred that dedicated energy beamforming is indeed beneficial when ID receivers possess the capability of cancelling the interference from energy signals, even with suboptimal designs.

\section{Conclusion}\label{sec:conclusions}

This paper has studied the joint information and energy transmit beamforming design for a multiuser MISO broadcast system for simultaneous wireless information and power transfer (SWIPT). The  weighted sum-power harvested by EH receivers is maximized subject to individual SINR constraints at ID receivers. Considering two types of ID receivers without or with the interference cancellation capability, the design  problems are formulated as two non-convex QCQPs, which are  solved optimally by applying the techniques of SDR and uplink-downlink duality. The results of this paper provide useful guidelines for practically optimizing the performance of multi-antenna SWIPT systems with receiver-location-based information and energy  transmission.

\appendix

\subsection{Proof of Proposition \ref{proposition:3.1}}\label{appendix:1}

Note that $\mathtt{(SDR1)}$ is a SDP problem and thus is convex. It is easy to verify that this problem satisfies the Slater's condition \cite{BoydVandenberghe2004} and thus has a zero duality gap. Therefore, we consider the Lagrangian of $\mathtt{(SDR1)}$ expressed as
\begin{align*}
&\mathcal{L}_1\left(\{\mv{W}_i\},\mv{W}_E,\{\lambda_i\},\beta\right) \\\nonumber=
& \beta P - \sum\limits_{i\in\mathcal{K_I}}\lambda_i \sigma_i^2 + \sum\limits_{i\in\mathcal{K_I}}\mathtt{tr}\left(\mv{A}_i\mv{W}_i\right) + \mathtt{tr}\left(\mv{C}_1\mv{W}_E\right),
\end{align*}
where
\begin{align}
\mv{A}_i ~ = ~&\mv{G}+\frac{\lambda_i\mv{h}_i^H\mv{h}_i}{\gamma_i}-\sum\limits_{k \neq i,k\in\mathcal{K_I}}\lambda_k\mv{h}_k^H\mv{h}_k - \beta\mv{I},\ \forall i \in \mathcal{K_I},\label{equation:app:1}\\
\mv{C}_1 ~ = ~&\mv{G}-\sum\limits_{k\in\mathcal{K_I}}\lambda_k\mv{h}_k^H\mv{h}_k - \beta\mv{I},\label{equation:app:2}
\end{align}
and $\lambda_i\ge 0, i\in\mathcal{K_I}$ and $\beta\ge 0$ are the dual variables associated with the $i$th SINR constraint and the transmit sum-power constraint of $\mathtt{(SDR1)}$, respectively. As a result, the dual problem of $\mathtt{(SDR1)}$ is given by
\begin{align*}
\mathtt{(SDR1.D)}:~\mathop \mathtt{min}_{\{\lambda_i\ge 0\}, \beta\ge 0} ~& \beta P - \sum\limits_{i\in\mathcal{K_I}}\lambda_i \sigma_i^2\\
\mathtt{s.t.}~~~~& \mv{C}_1 \preceq \mv{0}, \mv{A}_i \preceq \mv{0}, \forall i \in \mathcal{K_I}.
\end{align*}
Suppose that the optimal solution of $\mathtt{(SDR1.D)}$ is $\{\lambda_i^{\star}\}, \beta^{\star}$ and the resulting $\{\mv{A}_i\}, \mv{C}_1$ are $\{\mv{A}_i^{\star}\}, \mv{C}_1^{\star}$. Then the optimal solutions of $\mathtt{(SDR1)}$ and $\mathtt{(SDR1.D)}$ should satisfy the following complementary slackness conditions:
\begin{align}
\mathtt{tr}\left(\mv{A}_i^{\star}\mv{W}_i^{\star}\right) & = 0,\ \forall i \in \mathcal{K_I},\label{equation:app:3} \\
\mathtt{tr}\left(\mv{C}_1^{\star}\mv{W}_E^{\star}\right) & = 0,\label{equation:app:4}
\end{align}
which are equivalent to $\mv{A}_i^{\star}\mv{W}_i^{\star} = \mv{0}, \forall i \in \mathcal{K_I},$ and $\mv{C}_1^{\star}\mv{W}_E^{\star} = 0$, respectively. Furthermore, it can be verified that in order to meet the SINR constraints, it must hold that $\mv{W}_i^{\star} \neq \mv{0}$ or equivalently $\mathtt{rank}(\mv{W}_i^{\star})\ge 1, \forall i\in\mathcal{K_I}$, then from (\ref{equation:app:3}), it follows that $\mathtt{rank}(\mv{A}_i^{\star})\le M- 1, \forall i\in\mathcal{K_I}$.

Next, we prove this proposition by considering the following two cases where $\lambda_i^{\star} = 0, \forall i\in\mathcal{K_I}$, and (without loss of generality) there exists at least one $\bar i\in \mathcal{K_I}$ with  $\lambda_{\bar i}^{\star} > 0$, respectively.

First, we consider the case of $\lambda_i^{\star} = 0, \forall i\in\mathcal{K_I}$. In this case, we have $\mv{A}_i^\star = \mv{C}_1^\star = \mv{G}- \beta^\star\mv{I},\forall i\in\mathcal{K_I}$. Since $\mathtt{rank}(\mv{A}_i^{\star})\le M-1$ and $\mv{A}_i^\star \preceq \mv{0}, \forall i\in\mathcal{K_I}$, it follows that $\beta^\star = \xi_E$, where $\xi_E$ is the dominant eigenvalue of ${\mv G}$ (cf. (\ref{equa:jnl:7})). As a result, it can be verified from (\ref{equation:app:3}) and (\ref{equation:app:4}) that $\mv{W}_i^\star,\forall i\in\mathcal{K_I}$ and $\mv{W}_E^\star$ should all lie in the subspace spanned by ${\mv v}_E$, which can be shown to correspond to the OeBF-feasible case. Therefore, the case of $\lambda_i^{\star} = 0, \forall i\in\mathcal{K_I}$ cannot occur here.

Second, we consider the case when there exists at least one $\bar i\in \mathcal{K_I}$ with  $\lambda_{\bar i}^{\star} > 0$. In this case, we first show that any $\mv{W}_E^\star\succeq \mv{0}$ satisfying (\ref{equation:app:4}) should be zero. Given any $\mv{W}_E^\star\succeq \mv{0}$ satisfying (\ref{equation:app:4}), it follows that
\begin{align*}
&\lambda_i^\star(1+\frac{1}{\gamma_i})\mathtt{tr}\left( \mv{h}_i^H\mv{h}_i{\mv{W}}_E^\star\right)\\
=& \mathtt{tr}\left(\left(\mv{C}_1^\star + \lambda_i^\star(1+\frac{1}{\gamma_i})\mv{h}_i^H\mv{h}_i\right){\mv{W}}_E^\star\right)\\
 \le & \mathop\mathtt{max}_{\mv{X}\succeq \mv{0}}~\mathtt{tr}\left(\left(\mv{C}_1^\star + \lambda_i^\star(1+\frac{1}{\gamma_i})\mv{h}_i^H\mv{h}_i\right)\mv{X}\right)\\
 = &\mathop\mathtt{max}_{\mv{X}\succeq \mv{0}}~\mathtt{tr}(\mv{A}_i^\star\mv{X}) = 0,\ \forall i\in \mathcal{K_I},
\end{align*}
where the first equality follows from (\ref{equation:app:4}), the second inequality uses $\mv{W}_E^\star\succeq \mv{0}$, the third equality holds due to $\mv{A}_i^\star=\mv{C}_1^\star + \lambda_i^\star(1+\frac{1}{\gamma_i})\mv{h}_i^H\mv{h}_i$, and the last equality is true from the facts of $\mv{A}_i^\star \preceq \mv{0}$ and $\mathtt{rank}(\mv{A}_i^{\star})\le M-1, \forall i\in\mathcal{K_I}$.
Thus, it must hold that $\lambda_i^\star\mathtt{tr}\left( \mv{h}_i^H\mv{h}_i{\mv{W}}_E^\star\right)=0,\ \forall i\in \mathcal{K_I}$ or equivalently
$
\lambda_i^\star \mv{h}_i^H\mv{h}_i{\mv{W}}_E^\star=\mv{0},\ \forall i\in \mathcal{K_I}.$ As a result, we have
\begin{align}
\left(\mv{G}-\beta^\star\mv{I}\right){\mv{W}}_E^\star =&\left(\mv{G}-\sum\limits_{i\in\mathcal{K_I}}\lambda_i^\star\mv{h}_i^H\mv{h}_i - \beta^\star\mv{I}\right){\mv{W}}_E^\star\nonumber\\
=&\mv{C}_1^\star{\mv{W}}_E^\star=\mv{0},\label{equation:app:9:new9}
\end{align}
where the last two equalities hold due to (\ref{equation:app:2}) and (\ref{equation:app:4}), respectively. Furthermore, since $\lambda_{\bar i}^{\star} > 0$, $\bar i\in \mathcal{K_I}$, it follows that $\mv{h}_{\bar i}^H\mv{h}_{\bar i}{\mv{W}}_E^\star=\mv{0}$. Together with (\ref{equation:app:9:new9}), ${\mv{W}}_E^\star$ should lie in the null spaces of both $\mv{G}-\beta^\star\mv{I}$ and $\mv{h}_{\bar i}^H\mv{h}_{\bar i}$ at the same time. However, since the channel $\mv{g}_j$'s and $\mv{h}_i$'s are independently distributed {\color{black}under Assumption \ref{assumption1}}, we have $\mathtt{rank}(\mv{G}-\beta^\star\mv{I})\ge M-1$, and thus the two matrices $\mv{G}-\beta^\star\mv{I}$ and $\mv{h}_{\bar i}^H\mv{h}_{\bar i}$ span the entire space {\color{black}with probability one}. As a result, it follows that ${\mv{W}}_E^\star = \mv{0}$. Therefore, for any $\mv{W}_E^\star\succeq \mv{0}$ satisfying (\ref{equation:app:4}), it must hold that $\mv{W}_E^\star= \mv{0}$.

Finally, it remains to prove that $\mathtt{rank}({\mv{W}}_i^\star) = 1, \forall i\in\mathcal{K_I}$. We prove this result by showing that $\mathtt{rank}({\mv{A}}_i^\star) = M-1, \forall i\in\mathcal{K_I}$. By using $\mv{C}_1^\star = \mv{A}_{i}^\star - \lambda_{i}^\star(1+\frac{1}{\gamma_{i}})\mv{h}_{i}^H\mv{h}_{i}, \forall i\in\mathcal{K_I}$ together with the fact that $\mathtt{rank}(\mv{X}+\mv{Y}) \le \mathtt{rank}(\mv{X})+\mathtt{rank}(\mv{Y})$ holds for any two matrices $\mv{X}$ and $\mv{Y}$ of same dimension, it follows that $\mathtt{rank}({\mv{C}}_1^\star) \le \mathtt{rank}(\mv{A}_{i}^\star)+\mathtt{rank}\left( - \lambda_{i}^\star(1+\frac{1}{\gamma_{i}})\mv{h}_{i}^H\mv{h}_{i}\right), \forall i\in\mathcal{K_I}$. Given that any $\mv{W}_E^\star\succeq \mv{0}$ satisfying (\ref{equation:app:4}) should be zero, it can be shown that $\mathtt{rank}({\mv{C}}_1^\star) =M$; and meanwhile, $\mathtt{rank}\left(- \lambda_{i}^\star(1+\frac{1}{\gamma_{i}})\mv{h}_{i}^H\mv{h}_{i}\right) \le1, \forall i\in\mathcal{K_I}$. Therefore, we have $\mathtt{rank}(\mv{A}_{i}^\star) \ge \mathtt{rank}({\mv{C}}_1^\star) - \mathtt{rank}\left( - \lambda_{i}^\star(1+\frac{1}{\gamma_{i}})\mv{h}_{i}^H\mv{h}_{i}\right) \ge M-1, \forall i\in\mathcal{K_I}$. Combining this argument with $\mathtt{rank}(\mv{A}_{i}^\star) \le M-1, \forall i\in\mathcal{K_I}$, it follows that $\mathtt{rank}(\mv{A}_{i}^\star) = M-1, \forall i\in\mathcal{K_I}$. Accordingly, from (\ref{equation:app:3}) we have $\mathtt{rank}({\mv{W}}_i^\star) = 1, \forall i\in\mathcal{K_I}$. Proposition \ref{proposition:3.1} is thus proved.
%
%
%

\subsection{Proof of Proposition \ref{proposition:3.2}}\label{appendix:2}

Note that $\mathtt{(SDR2)}$ is a SDP problem and thus is convex. It is easy to verify that this problem satisfies the Slater's condition \cite{BoydVandenberghe2004} and thus has a zero duality gap. Therefore, we consider the Lagrangian of $\mathtt{(SDR2)}$ given by
\begin{align*}
&\mathcal{L}_2(\{\mv{W}_i\},\mv{W}_E,\{\lambda_i\},\beta) \nonumber \\=& \sum\limits_{i\in\mathcal{K_I}}\mathtt{tr}(\mv{A}_i\mv{W}_i)+\mathtt{tr}(\mv{C}_2\mv{W}_E)-\sum\limits_{i\in\mathcal{K_I}}\lambda_i\sigma_i^2+\beta P,
\end{align*}
where $\mv{A}_i$ is given by (\ref{equation:app:1}), $\mv C_2=\mv{G}-\beta \mv{I}$, and $\lambda_i \ge 0, i\in\mathcal{K_I}$ and $\beta\ge 0$ are the dual variables associated with the $i$th SINR constraint and the transmit sum-power constraint of $\mathtt{(SDR2)}$, respectively. The dual problem of $\mathtt{(SDR2)}$ can be expressed as
\begin{align}
\mathtt{(SDR2.D)}:\mathop{\mathtt{max}}_{\{\lambda_i\ge 0\},\beta\ge0} & ~ \sum\limits_{i\in\mathcal{K_I}}\lambda_i\sigma_i^2-\beta P \nonumber \\ \mathtt {s.t.} ~~~& ~ \mv{C}_2\preceq 0, ~\mv{A}_i \preceq 0, \ \forall i \in \mathcal{K_I}.\nonumber
\end{align}
Suppose that the optimal solution of $\mathtt{(SDR2.D)}$ is $\{\lambda_i^*\},\beta^*$, and the resulting $\{\mv{A}_i\},\mv{C}_2$ are $\{\mv{A}_i^*\},\mv{C}_2^*$. Then the optimal solutions of $\mathtt{(SDR2)}$ and $\mathtt{(SDR2.D)}$ should satisfy the following complementary slackness conditions:
\begin{align}
\mathtt{tr}(\mv{A}_i^*\mv{W}_i^*) ~& = 0, \forall i \in \mathcal{K_I}, \label{equation:app:11}\\
\mathtt{tr}(\mv{C}_2^*\mv{W}_E^*) ~& = 0,\label{equation:app:12}
\end{align}
which are equivalent to $\mv{A}_i^{*}\mv{W}_i^{*} = \mv{0}, \forall i \in \mathcal{K_I},$ and $\mv{C}_2^{*}\mv{W}_E^{*} = 0$. Note that in order to meet the SINR constraints, it must hold that $\mv{W}_i^{*} \neq \mv{0}$ or equivalently $\mathtt{rank}(\mv{W}_i^{*})\ge 1, \forall i\in\mathcal{K_I}$, then from (\ref{equation:app:11}) it follows that
\begin{align}\label{eqn:add:rank:A:SDR2}
\mathtt{rank}(\mv{A}_i^{*})\le M- 1, \forall i\in\mathcal{K_I}.
\end{align}

Next, we prove this proposition by focusing on the case when there exists at least one $\bar i\in \mathcal{K_I}$ with $\lambda_{\bar i}^{*} > 0$.{\footnote{Note that similarly to the proof given in Appendix \ref{appendix:1}, the case with $\lambda_i^{*} = 0, \forall i\in\mathcal{K_I}$ can be shown to correspond to the OeBF-feasible case and thus is not considered here.}} In this case, we first show $\mv{W}_E^*=q^*{\mv v}_E{\mv v}_E^H$ as follows. Due to the fact that $\mv{C}_2^* = \mv{G}- \beta^*\mv{I}\preceq \mv{0}$, we have $\beta^* \ge \xi_E$. If $\beta^* = \xi_E$, then $\mathtt{rank}(\mv{C}_2^*) = M-1$; it thus follows from (\ref{equation:app:12}) that $\mv{W}_E^*=q^*{\mv v}_E{\mv v}_E^H$ with $0\le q^* \le P$. If $\beta^* > \xi_E$, then $\mathtt{rank}(\mv{C}_2^*) = M$; it thus follows from (\ref{equation:app:12}) that $\mv{W}_E^* = \mv{0}$ or equivalently $\mv{W}_E^*=q^*{\mv v}_E{\mv v}_E^H$ with $q^* = 0$. Therefore, $\mv{W}_E^*=q^*{\mv v}_E{\mv v}_E^H$ and accordingly $\mathtt{rank}(\mv{W}_E^*)\le 1$ follows.

Second, we prove $\mathtt{rank}(\mv{W}_i^{*}) = 1$ by showing $\mathtt{rank}(\mv{A}_i^*) = M-1, \forall i\in\mathcal{K_I}$. To this end, we first prove that $\lambda_i^{*}, \forall i\in\mathcal{K_I}$ are all strictly positive by contradiction. Suppose that there exists one $\tilde i \in\mathcal{K_I}, \tilde i \neq \bar i$ satisfying that $\lambda_{\tilde i}^{*} = 0$. In this case, since $\left(-\sum\limits_{k \neq {\tilde i},k\in\mathcal{K_I}}\lambda_k^*\mv{h}_k^H\mv{h}_k\right) \preceq \mv{0}$, $\mv{C}_2^*\preceq \mv{0}$, $\mathtt{rank}\left(-\sum\limits_{k \neq {\tilde i},k\in\mathcal{K_I}}\lambda_k^*\mv{h}_k^H\mv{h}_k\right) \geq 1$ (due to $\lambda_{\bar i}^* > 0, \bar i \neq \tilde i$), and $\mathtt{rank}(\mv{C}_2^*) \ge M-1$, it can be shown that $\mathtt{rank}(\mv{A}_{\tilde i}^*) = \mathtt{rank}\left(\mv{C}_2^* -\sum\limits_{k \neq {\tilde i},k\in\mathcal{K_I}}\lambda_k^*\mv{h}_k^H\mv{h}_k\right) = M$ {\color{black}with probability one}, provided that the channel $\mv{g}_j$'s and $\mv{h}_i$'s are independently distributed {\color{black}under Assumption \ref{assumption1}}. This induces a contradiction to (\ref{eqn:add:rank:A:SDR2}). Therefore, the presumption cannot be true and it follows that $\lambda_i^{*} > 0, \forall i\in\mathcal{K_I}$.

With $\lambda_i^{*} > 0, \forall i\in\mathcal{K_I}$, it can be shown that $\mathtt{rank}\left(\mv{C}_2^* -\sum\limits_{k \neq i,k\in\mathcal{K_I}}\lambda_k^*\mv{h}_k^H\mv{h}_k\right) = M, \forall i\in\mathcal{K_I}$. By noting that $\mv{C}_2^* -\sum\limits_{k \neq i,k\in\mathcal{K_I}}\lambda_k^*\mv{h}_k^H\mv{h}_k = \mv{A}_i^*  - \frac{\lambda_i^*\mv{h}_i^H\mv{h}_i}{\gamma_i}$, we have $\mathtt{rank}(\mv{A}_i^{*})\ge \mathtt{rank}\left(\mv{C}_2^* -\sum\limits_{k \neq i,k\in\mathcal{K_I}}\lambda_k^*\mv{h}_k^H\mv{h}_k\right) - \mathtt{rank}\left(- \frac{\lambda_i^*\mv{h}_i^H\mv{h}_i}{\gamma_i}\right) \ge M- 1, \forall i\in\mathcal{K_I}$. Together with (\ref{eqn:add:rank:A:SDR2}), it follows that $\mathtt{rank}(\mv{A}_i^*) = M-1,\forall i\in\mathcal{K_I}$; accordingly $\mathtt{rank}(\mv{W}_i^{*}) = 1, \forall i\in\mathcal{K_I}$, holds due to (\ref{equation:app:11}). This thus completes the proof of Proposition \ref{proposition:3.2}.

\subsection{Proof of Proposition \ref{proposition:KI1}}\label{appendix:8}

Given $K_I=1$, let the optimal dual solution for $\mathtt{(SDR2)}$ be denoted by $\lambda_1^*\ge 0$ and $\beta^*\ge 0$ (see $\mathtt{(SDR2.D)}$ in Appendix \ref{appendix:2}). It then follows from (\ref{equation:app:11}) and (\ref{equation:app:12}) that the optimal solution of $\mathtt{(SDR2)}$ should satisfy the two equations of $\left(\mv{G}+\frac{\lambda_1^*\mv{h}_i^H\mv{h}_i}{\gamma_i}- \beta^*\mv{I}\right)\mv{W}_i^*= \mv{0}$ and $(\mv{G}- \beta^*\mv{I})\mv{W}_E^*= \mv{0}$ at the same time, where $\mv{G}+\frac{\lambda_1^*\mv{h}_i^H\mv{h}_i}{\gamma_i}- \beta^*\mv{I}\preceq\mv{0}$ and $\mv{G}- \beta^*\mv{I}\preceq\mv{0}$. Furthermore, we have $\beta^* \ge \xi_E$ due to $\mv{G}- \beta^*\mv{I}\preceq\mv{0}$.

Next, we prove this proposition by considering the two cases of $\beta^* = \xi_E$ and $\beta^* > \xi_E$, respectively. If $\beta^* = \xi_E$, then it can be shown that in order to satisfy  $\mv{G}+\frac{\lambda_1^*\mv{h}_i^H\mv{h}_i}{\gamma_i}- \beta^*\mv{I}\preceq\mv{0}$ and $\mv{G}- \beta^*\mv{I}\preceq\mv{0}$ at the same time, it must hold that $\lambda_1^* = 0$, which corresponds to the OeBF-feasible case. As a result, $\beta^* = \xi_E$ cannot occur here. On the other hand, if $\beta^* > \xi_E$, then $\mv{G}- \beta^*\mv{I}\prec\mv{0}$ is of full rank; accordingly, from $(\mv{G}- \beta^*\mv{I})\mv{W}_E^*= \mv{0}$ it follows that $\mv{W}_E^*= \mv{0}$. Proposition \ref{proposition:KI1} is thus proved.

\subsection{Proof of Proposition \ref{proposition:4.1}}\label{appendix:new:D:July17}
Denote the optimal values of ${\mathtt{(P1.1)}}$ and ${\mathtt{(P1.2)}}$ as $v_{\mathtt{(P1.1)}}$ and $v_{\mathtt{(P1.2)}}$, respectively. Since ${\mathtt{(P1.2)}}$ is the dual problem of ${\mathtt{(P1.1)}}$, it immediately follows that $v_{\mathtt{(P1.2)}} \ge v_{\mathtt{(P1.1)}}$. Therefore, to complete the proof of this proposition, we only need to show that $v_{\mathtt{(P1.1)}} \ge v_{\mathtt{(P1.2)}}$.

First, we express the SDR of problem ${\mathtt{(P1.1)}}$ as
\begin{align*}
&\mathtt{(SDR1.1)}:\\
\mathop{\mathtt{max}}_{\{\mv{W}_i\}}
& \sum\limits_{i\in\mathcal{K_I}}\mathtt{tr}(\mv{G}\mv{W}_i) \\
\mathtt {s.t.} &  \frac{\mathtt{tr}(\mv{h}_i^H\mv{h}_i\mv{W}_i)}{\gamma_i}-\sum\limits_{k\neq i,k\in\mathcal{K_I}}\mathtt{tr}(\mv{h}_i^H\mv{h}_i\mv{W}_k)-\sigma_i^2 \geq 0,  \forall i \in \mathcal{K_I}\\ &
\sum\limits_{i\in\mathcal{K_I}}\mathtt{tr}(\mv{W}_i)\leq P\\
& {{\mv W}_i}\succeq {\mv 0}, \forall i\in \mathcal{K_I}.
\end{align*}
From Proposition \ref{proposition:3.1}, it follows that ${\mathtt{(SDR1.1)}}$ always has a rank-one solution. Hence, by denoting the optimal value achieved by ${\mathtt{(SDR1.1)}}$ as $v_{\mathtt{(SDR1.1)}}$, we have $v_{\mathtt{(P1.1)}}=v_{\mathtt{(SDR1.1)}}$.

Meanwhile, we can express the SDR of problem (\ref{equa:jnl:8}) as
\begin{align}
&f_{\mathtt{SDR},1}(\beta) \triangleq~\nonumber\\
\mathop{\mathtt{max}}_{\{\mv{W}_i\}}
& \sum\limits_{i\in\mathcal{K_I}}\mathtt{tr}(\mv{G}\mv{W}_i) - \beta\left(
\sum\limits_{i\in\mathcal{K_I}}\mathtt{tr}(\mv{W}_i)- P\right)\nonumber \\
\mathtt {s.t.}~~ &  \frac{\mathtt{tr}(\mv{h}_i^H\mv{h}_i\mv{W}_i)}{\gamma_i}-\sum\limits_{k\neq i,k\in\mathcal{K_I}}\mathtt{tr}(\mv{h}_i^H\mv{h}_i\mv{W}_k)-\sigma_i^2 \geq 0,  \forall i \in \mathcal{K_I}\nonumber\\
~& {{\mv W}_i}\succeq {\mv 0}, \forall i\in \mathcal{K_I},\label{eqn:f_sdr1}
\end{align}
and accordingly define a new problem as
\begin{align*}
{\mathtt{(SDR1.2)}}:~\mathop \mathtt{min}\limits_{\beta \ge 0}~f_{\mathtt{SDR},1}(\beta).
\end{align*}
Then it is observed that ${\mathtt{(SDR1.2)}}$ is also the dual problem of ${\mathtt{(SDR1.1)}}$. Since ${\mathtt{(SDR1.1)}}$ is convex and satisfies the Slater's condition \cite{BoydVandenberghe2004}, it can be verified that strong duality holds between ${\mathtt{(SDR1.1)}}$ and ${\mathtt{(SDR1.2)}}$, i.e., $v_{\mathtt{(SDR1.1)}}=v_{\mathtt{(SDR1.2)}}$ with $v_{\mathtt{(SDR1.2)}}$ being the optimal value of ${\mathtt{(SDR1.2)}}$.  Together with $v_{\mathtt{(P1.1)}}=v_{\mathtt{(SDR1.1)}}$, we thus have $v_{\mathtt{(P1.1)}}=v_{\mathtt{(SDR1.2)}}$.

Moreover, since problem (\ref{eqn:f_sdr1}) to obtain $f_{\mathtt{SDR},1}(\beta)$ is the SDR of
problem (\ref{equa:jnl:8}) to obtain $f_{1}(\beta)$, we have $f_{\mathtt{SDR},1}(\beta)\ge f_{1}(\beta)$ for any $\beta \ge 0$. As a consequence, it follows that $v_{\mathtt{(SDR1.2)}} \ge v_{\mathtt{(P1.2)}}$. Hence, it is verified that $v_{\mathtt{(P1.1)}}=v_{\mathtt{(SDR1.2)}}\ge v_{\mathtt{(P1.2)}}$. Proposition \ref{proposition:KI1} is thus proved.

\subsection{Proof of Proposition \ref{theorem:macbc}}\label{appendix:D:April10}

First, we express the SDR of problem (\ref{equation:13}) as
\begin{align}
\mathop\mathtt{min}_{\{{\mv W}_i\succeq \mv{0}\}} ~& \sum\limits_{i\in\mathcal{K_I}}{\mathtt{tr}}\left((\beta\mv{I}-{\mv G}){\mv W}_i\right)\nonumber \\
{\mathtt{s.t.}}~~~& \frac{{\mathtt{tr}}({\mv h}_i^H{\mv h}_i{\mv W}_i)}{\gamma_i} - \sum\limits_{k\neq i,k\in\mathcal{K_{I}}}{\mathtt{tr}}({\mv h}_i^H{\mv h}_i{\mv W}_k) \nonumber \\&~~~~~~~~~~~~~~~~~~~~~~- \sigma_i^2 \ge 0, \forall i \in \mathcal{K_I},\label{equa:theorem1}
\end{align}
which can be shown to achieve the same optimal value of problem (\ref{equation:13}), for which the proof is similar to that of Proposition \ref{proposition:3.1} and thus is omitted here for brevity. We then express the dual problem of (\ref{equa:theorem1}) as
\begin{align}
    \mathop\mathtt{max}_{\{\tilde\lambda_i\ge 0\}}~&\sum\limits_{i\in\mathcal{K_I}}\tilde\lambda_i \sigma_i^2 ,\nonumber\\
    \mathtt{s.t.}~&-\frac{\tilde\lambda_i\mv{h}_i^H\mv{h}_i}{\gamma_i}+\sum\limits_{k \neq i,k\in\mathcal{K_I}}\tilde\lambda_k\mv{h}_k^H\mv{h}_k \nonumber \\&~~~~~~~~~~~~~~~~~~~+ \beta\mv{I}-\mv{G}\succeq \mv{0},\ \forall i \in \mathcal{K_I},\label{equa:51}
\end{align}
where $\tilde\lambda_i\ge 0, i \in \mathcal{K_I}$ represents the dual variable associated with the $i$th constraint in (\ref{equa:theorem1}). Since problem (\ref{equa:theorem1}) is convex and satisfies the Slater's condition \cite{BoydVandenberghe2004}, strong duality holds between (\ref{equa:theorem1}) and (\ref{equa:51}). As a result, given $g(\beta)>-\infty$, (\ref{equa:51}) is always feasible, and problems (\ref{equation:13}), (\ref{equa:theorem1}) and (\ref{equa:51}) achieve the same optimal value.

We then present the following lemma, which is proved in \cite[Lemma 1]{YuLan2007}.
\begin{lemma}\label{lemma:a1}
Let $\mv{A}$ be an $n\times n$ symmetric positive semidefinite matrix and $\mv{b}$ be an $n\times 1$ vector. Then, $\mv{A}\succeq \mv{b}\mv{b}^H$ if and only $\mv{b}^H\mv{A}^\dagger \mv{b}^H\le 1$.
\end{lemma}

From Lemma \ref{lemma:a1}, it then follows that by setting $\mv{A}=\sum\limits_{k\neq i,k\in\mathcal{K_I}}\tilde\lambda_k\mv{h}_k^H\mv{h}_k + \beta\mv{I}-\mv{G}$ and $\mv{b}=\sqrt{\frac{\tilde\lambda_i}{\gamma_i}}\mv{h}_i^H$ the SINR constraints of (\ref{equa:51}) are identical to
\begin{align}\label{euqa:theorem:4}
\frac{\tilde\lambda_i}{\gamma_i}\mv{h}_i\left(\sum\limits_{k\neq i,k\in\mathcal{K_I}}\tilde\lambda_k\mv{h}_k^H\mv{h}_k + \beta\mv{I}-\mv{G}\right)^{\dagger}\mv{h}_i^H \le 1, \ \forall i \in \mathcal{K_I}.\end{align}
Note that the inequalities in (\ref{euqa:theorem:4}) are equivalent to the following inequalities: \begin{align}\label{euqa:theorem:5}\mathop{\mathtt{max}}_{\{\tilde{\mv{w}}_i\}}{\mathtt{SINR}}_i^{\mathrm{MAC}}(\{\tilde{\mv{w}}_i,\tilde\lambda_i\}) \leq \gamma_i, \forall i \in \mathcal{K_I},\end{align}
which can be easily verified by substituting the optimal solution of the problem in the left-hand side (LHS) of (\ref{euqa:theorem:5}), i.e., $\tilde{{\mv{w}}}_i = \left(\sum_{k\neq i,k\in\mathcal{K_I}}\tilde\lambda_k\mv{h}_k^H\mv{h}_k + \beta\mv{I}-\mv{G}\right)^{\dagger}\mv{h}_i^H$, into (\ref{euqa:theorem:5}). As a result, problem (\ref{equa:51}) and hence problem (\ref{equation:13}) are equivalent to the following problem:
\begin{align}\mathop{\mathtt{max}}_{\{\tilde\lambda_i\ge 0\}}
& ~\sum\limits_{i\in\mathcal{K_I}}\tilde\lambda_i\sigma_i^2 \nonumber\\
\mathtt {s.t.} & ~ \mathop{\mathtt{max}}_{\{\tilde{\mv{w}}_i\}}{\mathtt{SINR}}_i^{\mathrm{MAC}}(\{\tilde{\mv{w}}_k,\tilde\lambda_k\}) \le \gamma_i, \ \forall i \in \mathcal{K_I}.\label{equa:90}
\end{align}

Next, problem (\ref{equa:gd}) can be equivalently re-expressed as
\begin{align}\mathop{\mathtt{min}}_{\{\lambda_i\ge 0\}}
& ~\sum\limits_{i\in\mathcal{K_I}}\lambda_i\sigma_i^2 \nonumber\\
\mathtt {s.t.} & ~ \mathop{\mathtt{max}}_{\{\tilde{\mv{w}}_i\}}{\mathtt{SINR}}_i^{\mathrm{MAC}}(\{\tilde{\mv{w}}_i,\lambda_i\}) \geq \gamma_i, \ \forall i \in \mathcal{K_I}.\label{equa:999}
\end{align}
By comparing problems (\ref{equa:90}) and (\ref{equa:999}), their difference lies in the reversed SINR constraints as well as the reversed objective (maximization in (\ref{equa:90}) versus minimization in (\ref{equa:999})). It is evident that the optimal solution of both problems (\ref{equa:90}) and (\ref{equa:999}) should be obtained with all the SINR constraints being tight, since otherwise the objective value in (\ref{equa:90}) (in (\ref{equa:999})) can be further increased (decreased) by increasing (decreasing) certain $\tilde\lambda_i$'s ($\lambda_i$'s) while satisfying all the SINR constraints, provided that both the numerator and denominator of the expression of ${\mathtt{SINR}}_i^{\mathrm{MAC}}(\{\tilde{\mv{w}}_i,\lambda_i\})$ given in (\ref{eqn:sinr:uplink}) are non-negative. We can thereby show that problems (\ref{equa:90}) and (\ref{equa:999}) achieve the same optimal value as follows.

Note that the SINR constraints in (\ref{equa:90}) can be re-expressed as
\begin{align}\label{eqn:proof:convex}
&\lambda_i-\gamma_i\left(\mathop\mathtt{min}\limits_{\|\tilde{\mv{w}}_i\|=1}\frac{\tilde{\mv{w}}_i^H\left(\sum\limits_{k\neq i,k\in\mathcal{K_I}}\lambda_k\mv{h}_k^H\mv{h}_k+\beta\mv{I}-\mv{G}\right)\tilde{\mv{w}}_i}{\tilde{\mv{w}}_i^H\mv{h}_i^H\mv{h}_i\tilde{\mv{w}}_i}\right)\nonumber\\&~~~~~~~~~~~~~~~~~~~~~~~~~~~~~~~~~~~~~~~~~~~~
\leq 0, \ \forall i \in \mathcal{K_I},
\end{align}
which specifies a convex set of $\{\lambda_i\}$ due to the fact that the minimum of a linear function is concave \cite{SchubertBoche2012}. In other words, problem (\ref{equa:90}) is recast as a convex problem. As a result, any local optimum point of (\ref{equa:90}) is globally optimal. Therefore, it can be shown that any $\{\tilde\lambda_i\}$ in (\ref{equa:90}) that makes all SINR constraints being tight is globally optimal and thus achieves the optimal value for (\ref{equa:90}). Hence, it follows that problems (\ref{equa:90}) and (\ref{equa:999}) achieve the same optimal value. Consequently, problems (\ref{equation:13}) and (\ref{equa:gd}) also achieve the same optimal value, which completes the proof of Proposition \ref{theorem:macbc}.

\begin{remark}\label{remakr:a:1}
It is worth pointing out that the proof of the uplink-downlink duality for the case of $\beta\mv{I}-\mv{G} \nsucceq \mv{0}$ and $g(\beta)>-\infty$ in Proposition \ref{theorem:macbc} differs from that for the case of $\beta\mv{I}-\mv{G} \succeq \mv{0}$ (see \cite[Theorem 1]{YuLan2007}) in two main aspects. First, to show the equivalent relationship between (\ref{equation:13}) and (\ref{equa:90}) in the case of $\beta\mv{I}-\mv{G} \nsucceq \mv{0}$ and $g(\beta)>-\infty$, we use the dual problem of the SDR problem in (\ref{equation:13}), i.e., problem (\ref{equa:theorem1}), given the fact that the SDR of (\ref{equation:13}) is tight; whereas in the case of $\beta\mv{I}-\mv{G} \succeq \mv{0}$, \cite{YuLan2007} uses the SOCP reformulation of (\ref{equation:13}) to show that strong duality holds for (\ref{equation:13}) and its dual problem. Second, to show that problems (\ref{equa:90}) and (\ref{equa:999}) achieve the same optimal value in the case of $\beta\mv{I}-\mv{G} \nsucceq \mv{0}$ and $g(\beta)>-\infty$, we use the fact that problem (\ref{equa:90}) can be recast as a convex problem; whereas in the case of $\beta\mv{I}-\mv{G} \succeq \mv{0}$, \cite{YuLan2007} uses the technique of standard interference function \cite{Yates1995}. Moreover, note that our proof for Proposition \ref{theorem:macbc} is also applicable to the case of $\beta\mv{I}-\mv{G} \succeq \mv{0}$; however, the proof in \cite{YuLan2007} cannot be applied for the case of $\beta\mv{I}-\mv{G} \nsucceq \mv{0}$ and $g(\beta)>-\infty$ to obtain Proposition \ref{theorem:macbc} here.
\end{remark}

\subsection{Proof of Proposition \ref{proposition:July19:2}} \label{appendix:4}

The first part of this proposition follows directly from the proof of Proposition \ref{theorem:macbc} in Appendix \ref{appendix:D:April10}. Thus, we only need to prove the second part of this proposition as follows.


From Appendix \ref{appendix:D:April10}, it follows that due to the equivalence between problems (\ref{equa:90}) and (\ref{equa:999}) the dual variables $\{\tilde\lambda_i\}$ in problem (\ref{equa:51}) and the uplink transmit power $\{\lambda_i\}$ in problem (\ref{equa:gd}) are identical. As a result, given that problem (\ref{equa:51}) is feasible in the case of $g(\beta) > -\infty$, the optimal uplink transmit power $\{\lambda_i^\star\}$ must be a feasible solution of (\ref{equa:51}) and thus satisfy the constraints in (\ref{equa:51}). Accordingly, the second part of the proposition follows. Therefore, Proposition \ref{proposition:July19:2} is proved.

\subsection{Proof of Proposition \ref{theorem:converge}}\label{appendix:7}
For the fixed point iteration with the proposed initial point, {\color{black}we have} shown that $\{\lambda_i^{(n)}\}$ is an element-wise monotonically decreasing sequence and at the same time lower bounded by $\{\lambda_i^{\star}\}$; as a result, the fixed point iteration will converge to a stationary point for problem (\ref{equa:gd}). Note that problem (\ref{equa:gd}) is identical to problem (\ref{equa:90}), which can be recast as a convex problem (cf. (\ref{eqn:proof:convex})). Therefore, any stationary point of (\ref{equa:90}) as well as (\ref{equa:gd}) should be globally optimal. Accordingly, the converged stationary point is the optimal solution $\{\lambda_i^{\star}\}$. Proposition \ref{theorem:converge} is thus proved.

{\color{black}
\begin{remark}
It is worth noting that the optimality of $\{\lambda_i^{\star}\}$ in Proposition \ref{theorem:converge} cannot be shown by using the techniques of standard \cite{Yates1995} or general \cite{SchubertBoche2012} interference functions that are normally utilized in existing uplink-downlink duality related literature (see e.g. \cite{YuLan2007}). This is because that the equivalent noise term ${\tilde{\mv{w}}_i^{\star H}\left(\beta\mv{I}-\mv{G}\right)\tilde{\mv{w}}_i^\star}$ in our case can be negative for some $i \in \mathcal{K_I}$ due to $\beta\mv{I}-\mv{G} \nsucceq \mv{0}$. The same challenge exists when we prove the uplink-downlink duality in Proposition \ref{theorem:macbc} (see Remark \ref{remakr:a:1}).
\end{remark}
}

\subsection{Proof of Proposition \ref{proposition:4.5}}\label{appendix:6:add}
\begin{lemma}\label{lemma:a2}
If there exists any one fixed point $\{\lambda_i^{\star}\}$ with $\lambda_i^{\star} \ge 0, \forall i\in\mathcal{K_I}$ satisfying $\lambda_i^{\star}=\mathtt{m}_i\left(\{\lambda_i^{\star}\}\right), \forall i \in \mathcal{K_I}$, then it follows that $g(\beta)>-\infty$.
\end{lemma}\begin{proof}
By substituting the optimal MMSE filter in (\ref{eqn:iterative:1}) into the equations $\lambda_i^{\star}=\mathtt{m}_i\left(\{\lambda_i^{\star}\}\right), \forall i \in \mathcal{K_I}$, it follows that $\frac{\lambda_i^{\star}}{\gamma_i}\mv{h}_i\left(\sum\limits_{k\neq i,k\in\mathcal{K_I}}\lambda_k^{\star}\mv{h}_k^H\mv{h}_k + \beta\mv{I}-\mv{G}\right)^{\dagger}\mv{h}_i^H = 1, \ \forall i \in \mathcal{K_I}$. Then based on Lemma \ref{lemma:a1}, we have $\frac{\lambda_i^{\star}\mv{h}_i^H\mv{h}_i}{\gamma_i}+\sum\limits_{k \neq i,k\in\mathcal{K_I}}\lambda_k^{\star}\mv{h}_k^H\mv{h}_k + \beta\mv{I}-\mv{G}\succeq \mv{0},\ \forall i \in \mathcal{K_I}$. As a result, the fixed point $\{\lambda_i^{\star} \}$ is a feasible solution of (\ref{equa:51}), which makes the optimal value achieved by (\ref{equa:51}) bounded from below. Together with the strong duality between (\ref{equa:theorem1}) and (\ref{equa:51}), it follows that $g(\beta)>-\infty$. Lemma \ref{lemma:a2} is thus proved.
\end{proof}

From Lemma \ref{lemma:a2}, it is evident that if $g(\beta)=-\infty$, then there does not exist any fixed point of $\lambda_i^{\star} \ge 0, \forall i\in\mathcal{K_I}$ satisfying $\lambda_i^{\star}=\mathtt{m}_i\left(\{\lambda_i^{\star}\}\right), \forall i \in \mathcal{K_I}$. In this case, since the fixed point iteration leads to an element-wise monotonically decreasing sequence of  $\{\lambda_i^{(n)}\}$, it will always result in a sufficiently small $\{\lambda_i^{(n)}\}$ with $\sum\limits_{k \neq i,k\in\mathcal{K_I}}\lambda_k^{(n)}\mv{h}_k^H\mv{h}_k + \beta\mv{I}-\mv{G}\nsucceq \mv{0}$ for some $i\in\mathcal{K_I}$, given that $\beta\mv{I}-\mv{G}\nsucceq \mv{0}$. Thus, Proposition \ref{proposition:4.5} is proved.

\end{document}